\documentclass[11pt,letterpaper]{article}

\usepackage[margin=1in]{geometry}
\usepackage[T1]{fontenc}
\usepackage{lmodern}
\usepackage{booktabs}
\usepackage{amsmath,amssymb,amsthm,mathtools}
\usepackage{enumitem}
\usepackage{needspace}
\usepackage{xcolor}
\usepackage{tikz}
\usepackage{wrapfig}
\usepackage[font=small]{caption}
\usepackage{microtype}
\usepackage[numbers,square]{natbib}
\usepackage[colorlinks=true,
  linkcolor=pink!60!blue,
  citecolor=pink!60!blue,
  urlcolor=pink!60!blue]{hyperref}

\hypersetup{
  pdftitle={Local Search with Correlated Randomness},
  pdfauthor={Yunbei Xu}
}

\newtheorem{theorem}{Theorem}[section]
\newtheorem{proposition}[theorem]{Proposition}
\newtheorem{corollary}[theorem]{Corollary}
\newtheorem{lemma}[theorem]{Lemma}
\theoremstyle{definition}

\theoremstyle{remark}

\AddToHook{env/theorem/before}{\Needspace{5\baselineskip}}
\AddToHook{env/proposition/before}{\Needspace{5\baselineskip}}
\AddToHook{env/corollary/before}{\Needspace{5\baselineskip}}
\AddToHook{env/lemma/before}{\Needspace{5\baselineskip}}

\newcommand{\Ncal}{\mathcal N}
\newcommand{\Prb}{\mathbb P}
\newcommand{\E}{\mathbb E}
\newcommand{\1}{\mathbf 1}

\title{Local Search with Correlated Randomness}
\author{Yunbei Xu\\National University of Singapore\\\texttt{yunbei@nus.edu.sg}}
\date{}

\begin{document}
\hypersetup{pageanchor=false}
\maketitle
\thispagestyle{empty}
\begin{abstract}
The running-time distribution of a randomized algorithm is computed
under independent randomness, but independence is rarely what a source
guarantees. We study sources $\nu$ with $\nu[w]\le DP[w]^s$ for every
finite prefix $w$ of the random tape, where $P$ is an independent
reference law, $0<s\le1$, and $D\ge1$; individual choices may be
predictable, even determined by the past. The worst-case probability of
any closed tape event $E$ is exactly $\min\{1,D\mathcal C_s(E)\}$, where
$\mathcal C_s(E)$ is the infimum cost of covering $E$ by prefixes of
varying depth. This tree max-flow/min-cut identity makes worst-case
probabilities exact and optimizable over the algorithm's choices.

Three applications follow. In scheduling, $n$ tasks share
one machine and each is retried until the next random symbol lies in its
success set; every schedule has the same makespan distribution under
independent randomness, yet under our constraint two fixed orders miss a
linear deadline with worst-case probabilities $1$ and
$\exp[-\Theta(n)]$, and a static priority rule is optimal at every
deadline. For the standard local walk on a $k$-CNF with an assignment
satisfying a majority of each clause, we determine the sharp source
threshold $s_*$: below it, on some instances a single admissible source
forces nontermination however clauses are selected; at it, with fixed
$D$, the expected number of flips is at most cubic in the initial
Hamming distance, against linear under independence, with matching
instances at every distance up to a constant fraction of the number of
variables; and any fixed margin above it restores linear dependence on
the distance, with a constant inversely proportional to the margin. For
noisy predecessor search under a prefix bound on the comparison errors,
verified backtracking achieves any prescribed error probability at or
above a source threshold whose sharpness we leave open, with the same
cubic-to-linear query bounds. In every setting the algorithm may also restart, which supplies
no fresh randomness on a shared tape; restarting restores exponential
tails exactly when the nontermination event has a prefix cover of cost
below one.
\end{abstract}

\clearpage
\pagenumbering{roman}
\hypersetup{pageanchor=true}
\setcounter{tocdepth}{2}
\pdfbookmark[1]{Contents}{contents}
\tableofcontents
\clearpage
\pagenumbering{arabic}

\section{Introduction}\label{sec:localization-motivation}
How much does an algorithm's running-time distribution under independent
randomness reveal about its behavior when independence is no longer
guaranteed? The distribution gives the probability that the algorithm
is still running at each time, but not which early choices the
surviving runs share. When later choices can depend on earlier ones,
that missing structure matters, and it can matter decisively. This
paper studies the question for an arbitrary algorithm whose only
randomness comes from a single tape. We impose one cumulative
 constraint on the joint law of the algorithm's random choices: every
initial segment of the tape has probability at most a constant times a
fixed power of its probability under an independent reference law.

The source can be viewed as an adversary that sees the algorithm's
state and every choice made so far, limited by one budget for the
whole run rather than by a separate constraint at each step. We study
this general source model in three settings, scheduling,
local SAT search, and noisy search, and for one decision available in
all of them, when to restart. In each, the worst case over admissible sources can be
computed exactly, and the answer changes what the algorithm should do:
in scheduling it picks the order in which to attempt tasks, an order
the running-time law under independent randomness cannot see; in local
SAT search it gives a sharp source threshold, with worst-case time
cubic rather than linear at the threshold; in noisy search it says how
to verify answers under correlated errors; and in every setting it says
exactly when restarting on the same tape restores exponential tails.

\paragraph{The source model.}
An algorithm consumes a tape of primitive random choices, one symbol at
a time, from a finite alphabet. A prefix $w=a_1\cdots a_t$
consists of the first $t$ symbols of the tape, that is, of all random
choices made through time $t$, and for a law $\nu$ on tapes, $\nu[w]$
is the probability that the tape begins with $w$. Let $P$ be the
reference law under which the symbols are independent with a common
per-symbol law. We allow the tape to be generated by any law $\nu$
satisfying the prefix probability bound
\[
 \nu[w]\le DP[w]^s\quad\text{for every finite prefix }w,
 \qquad 0<s\le1,\quad D\ge1.
\]
The exponent $s$ measures what fraction of the reference randomness is
guaranteed, and the constant $D$ allows a fixed amount of slack. For a
uniform $q$-symbol reference, the condition says exactly that the first
$t$ symbols have min-entropy at least $st\log_2q-\log_2D$ for every $t$,
where the min-entropy of a distribution is the negative base-two
logarithm of its largest outcome probability, the standard measure of
weak randomness \citep{chor-goldreich-1988}. Importantly, the
assumption constrains the tape as a whole, bounding the entire sequence
of choices up to each time and conditioning on nothing. Within a run, a
source may therefore spend randomness early and then produce a run of
symbols determined by what came before; across runs, the guarantee is
never renewed, so a restart gives the algorithm no fresh randomness.
The independent reference law prices these deviations and is the
baseline against which the results are compared; $s$ and $D$ set the
room for deviation, and the whole-prefix form decides how that room can
be used (Section~\ref{sec:prefix-cover-robustness}).

\paragraph{Why this class.}
Three properties make the class natural. First, it is general. A Chor--Goldreich source \citep{chor-goldreich-1988}, a
sequence of blocks each with min-entropy conditioned on all previous
ones, renews its guarantee at every block, and the almost
Chor--Goldreich sources of Doron, Moshkovitz, Oh, and Zuckerman
\citep{doron-et-al-2023} relax that guarantee, in their basic form by
letting a bounded number of blocks, at positions fixed in advance, be
arbitrary. A
Chor--Goldreich bound implies our condition, with $D$ absorbing the
block length, but our condition is a cumulative information constraint
on whole prefixes, not a per-block one: a run of determined symbols
may begin at a history-dependent time and last as long as the
randomness spent before it pays for. The adversarial source of
Section~\ref{sec:repair-order}, which forces zeros from the random
moment the $n$th zero appears, is such a run, and so is the SAT index
stream below, which repeats one index $m$ times after $18m$ random
ones. Second, it is exact: for every closed failure event, the
worst-case probability over the class is attained and equals the
cheapest cover of the event by prefixes of varying depth, a tree
max-flow/min-cut identity (Section~\ref{sec:prefix-cover-robustness})
that makes thresholds sharp and lets choices be optimized against the
worst case itself. Third, it is global: the guarantee is not inherited
by the subsequence that one component of an interleaved computation
sees, nor renewed by a restart, so restarts, interleavings, and replays
are analyzed, in all our theorems, under the single constraint on the
original tape.

\paragraph{Scheduling retry tasks.}
The simplest setting is a scheduling problem: $n$ \emph{retry tasks}
share one machine, each task is completed when the next random symbol
falls in its success set and otherwise remains unfinished, and a
\emph{selector}, our word for any algorithm's decision rule, here the
schedule, decides which unfinished task to attempt next, possibly using
everything observed so far. Under independent
symbols this is single-machine stochastic scheduling with independent
geometric processing times and a makespan deadline
\citep{pinedo-2022}, and every selector has the same makespan
distribution, because the makespan is the sum of the tasks' geometric
numbers of attempts. An attempt repairs one task and never undoes
another, which makes this the interference-free case of the walks that
follow. Our question is what happens when the processing times are
driven by one correlated tape.
For an explicit reference law
and source strength, two fixed orders miss a deadline linear in $n$
with worst-case probabilities $1$ and $\exp[-\Theta(n)]$. We then
identify, for each source strength, a fixed order that is optimal at
every deadline among all selectors, including history-dependent ones.

\paragraph{Local SAT search.}
The second setting is the standard local walk for SAT, Papadimitriou's
walk \citep{papadimitriou-1991} and the inner loop of Sch\"oning's
algorithm \citep{schoning-1999}. Tompkins and Hoos
\citep{tompkins-hoos-2006} found such walks robust to several forms of
derandomization but vulnerable to cyclic random streams; the following
stream illustrates what our assumption permits.
Consider a 3-CNF formula for which some assignment makes at least two
literals true in every clause; the walk is not told this assignment.
Repeatedly select a violated clause and read an index from $\{1,2,3\}$
to choose which of its three variables to flip. A violated clause has
all three literals false, so at least two of its variables disagree
with the promised assignment, and at least two of the three choices
move one step closer to it in Hamming distance. Under independent
uniform indices, this drift finds a satisfying assignment in expected
time linear in the initial distance
\citep{austrin-guruswami-hastad-2017}. Initialization and clause
selection can be deterministic.

Now fix a positive integer $m$ and generate the indices as follows:
$18m$ independent uniform indices, then the last of these repeated $m$
times, then independent uniform indices forever. Once the first $18m$
indices have been observed, the next $m$ indices are known. Every possible
prefix of length $t$ contains at least $18t/19$ independent draws, so
its probability is at most $3^{-18t/19}$: the stream satisfies our
assumption with $s=18/19$ and $D=1$. The number $18$ is the smallest
for which this exponent exceeds the source threshold of this walk,
which Section~\ref{sec:promise-sat} determines; our SAT theorem
therefore guarantees termination, with an exponential tail, for this
stream. Note that
restarting the walk during the repeated block changes the assignment
but does not restore fresh randomness: the next index is still known.
Section~\ref{sec:prefix-restarts} asks when restarting can help at all.

\paragraph{Noisy search.}
In the third setting the algorithm must not only finish but return the right answer. The problem is to find the predecessor of a target among sorted keys by comparisons, in which each comparison (a noisy label in active-learning terms) may be answered wrongly. The walk is on a binary search tree over the key
intervals. At each node it asks, by majority over repeated queries,
whether the target lies in the node's interval and on which side of the
split; it descends to the chosen child when the check passes and backs
up to the parent when it fails. The tape holds one error bit per
comparison. Under the reference law the errors are independent with a
fixed probability below $1/2$, and this verified backtracking
\citep{feige-et-al-1994-noisy} reaches the correct leaf after
$O(\log n)$ queries with any fixed error probability. Under our
condition a run
of wrong answers can be certain given the past; we ask what the same
walk guarantees under one correlated tape.

\paragraph{Approach.}
It helps to take the source's point of view. The source is a flow of
probability mass down the tree of prefixes, and the constraint caps
the mass through the node $w$ at $DP[w]^s$. The source wins by steering
mass into executions that fail, and the cheap way to do so is to bank
randomness early and spend it later by forcing symbols when the
algorithm is vulnerable. The tool is a \emph{prefix cover}: a finite collection
of prefixes such that every execution that has not succeeded by the
deadline begins with one of them. The analysis may stop at a short
prefix on one branch of the execution tree and follow another branch
much further. The cost of a cover is $\sum_wP[w]^s$; since
$\nu[w]\le DP[w]^s$, $D$ times the cost bounds the worst-case failure
probability.

The converse is the key fact: by the max-flow/min-cut theorem on the
prefix tree, a cover is a cut, and the worst-case failure probability
over all admissible sources is \emph{exactly} $\min\{1,D\mathcal C_s(E)\}$,
where $\mathcal C_s(E)$ is the infimum of the cover costs of the
closed failure event $E$, an infimum that is attained when $E$ is
determined by finitely many symbols. The
exactness is what makes optimization possible: minimizing cover cost
over the algorithm's choices is the same as minimizing worst-case
failure, and the minimum satisfies a Bellman recursion in which every
value is capped at one, the normalized cost of covering the current
prefix outright. All upper bounds in this paper are covers of this
kind, and all lower bounds are sources built from the max-flow side of
the identity.

\paragraph{What is technically new.}
The identity and the recursion are classical; the contributions are
what they are made to yield. The first is an exchange inequality
compatible with the truncation at one, which reduces a dynamic program
over $2^n$ task subsets to the $n$ suffixes of a sorted list
(Theorem~\ref{thm:retry-priority}). The second is a barrier lemma for
integer potentials under a worst-case count of favorable symbols
(Lemma~\ref{lem:promise-potential-barrier}), and the third a discrete
capped plateau with near-critical slack, from which one source forcing
the matching delay for every selector is built
(Lemma~\ref{lem:promise-moving-plateau},
Theorem~\ref{thm:promise-distance-sharp}). The fourth is the
impossibility half of the restart criterion, against all
history-dependent policies (Theorem~\ref{thm:prefix-restart-criterion}).

\paragraph{An optimal retry order that running-time laws cannot identify.}
For a uniform reference source, the optimal rule is to
finish a task with the fewest successful symbols first, that is,
largest expected length first; for a general reference law,
Section~\ref{sec:repair-order} replaces this count by a weight of
the success set. The source stalls a task by forcing symbols outside
its success set, and forcing a common symbol is cheap under the prefix
bound while forcing a rare one is expensive; so the task that needs a
rare symbol is the one the source can stall most cheaply, and it should
be attempted before the source has banked slack from earlier, more
random, symbols. Always attempting an unfinished task of minimum weight
minimizes worst-case failure at every deadline and every $D$, over all
deterministic history-dependent selectors
(Theorem~\ref{thm:retry-priority}). Neither the running-time law under
$P$ nor the sum of $s$-th powers of the surviving-prefix probabilities
at any single depth can distinguish these selectors: both statistics
are the same for all of them. The opening exponential separation between the best and worst
orders is Theorem~\ref{thm:repair-order-multiscale}; its proof is a
cover that cuts some executions at time $4n$ and others at time $8n$.

\paragraph{A sharp SAT threshold and matching delay.}
Take the walk above on a $k$-CNF for which some assignment makes at
least $r$ literals true in every clause, $k/2<r<k$. Under
independent uniform choices it takes $O_{k,r}(h)$ expected flips from
Hamming distance $h$. Under the prefix
bound, the sharp source threshold over this promise class is
$s_*=\log_k(2\sqrt{r(k-r)})$. At and above it,
\[
 \E_\nu\tau=O_{k,r}\!\left(\min\{L^3,L/(s-s_*)\}\right),
 \qquad L=h+\log D+1,\quad L/0=\infty,
\]
and at the threshold the probability of surviving $T$ flips is at most
$\exp[-\Omega_{k,r}(T^{1/3})]$ for fixed $h$ and $D$, a
stretched-exponential bound that a complete formula family matches. The
upper bound allows arbitrary clause overlap and history-dependent clause
selection. On suitable formulas and initial assignments, one source
forces the matching delay with probability one, simultaneously for
every selector. Thus, at the threshold and fixed $D$, a walk whose
independent-source expectation is linear in the initial distance $h$
can be forced to take time cubic in $h$ (Theorems~\ref{thm:promise-robust}
and~\ref{thm:promise-distance-sharp}), also when $h$ is a constant
fraction of the number of variables $n$: the expected number of flips
is then $\Theta(n^3)$ against $\Theta(n)$ under independence
(Corollary~\ref{cor:promise-linear-distance}). A fixed margin above the
threshold restores linear dependence on $L$. Below $s_*$,
nontermination is possible even with $D=1$
(Theorem~\ref{thm:promise-sharp}).

\paragraph{Correct output as well as termination in noisy search.}
Under the prefix bound on comparison errors, at any source
strength at or above a threshold $s_c$ that we show suffices but do not
show is sharp, we give a predecessor-search algorithm with error probability
at most $\eta$ and $O_p((h+\log(D/\eta)+1)^3)$ expected primitive
queries,
where $h$ is the correct leaf's depth and the algorithm needs neither
$h$ nor the tree height.
Longer verification chains at deeper leaves place all wrong outputs
beyond one common distance barrier, so a single cover controls the
total error without a union bound over leaves. A fixed margin above
this sufficient threshold gives linear dependence on the same quantity
(Theorem~\ref{thm:noisy-predecessor}).

\paragraph{An exact criterion for restarting on the same tape.}
A restart changes the state but supplies no fresh randomness.
Nevertheless, for a fixed algorithm, restarting can give exponential
tails exactly when its nontermination event has a prefix cover of cost
$c<1$, and the same condition characterizes whether any restart rule
can terminate almost surely against every source with $D=1$
(Theorem~\ref{thm:prefix-restart-criterion}): the cover tells the
algorithm which attempts to abandon, and $j$ abandoned attempts in a
row form a cover of cost $c^j$. For explicit finite-state processes,
two linear programs, one for the values and one that breaks ties among
value-optimal updates, compute a stationary selector minimizing
worst-case nontermination and a restart rule whenever robust
termination is possible (Theorem~\ref{thm:finite-state-restart-lp}).
Adaptivity can be necessary: a three-state process has worst-case
nontermination probability at least $1/2$ under every predetermined
cutoff sequence, while one stationary rule gives an exact exponential
tail (Proposition~\ref{prop:restart-oblivious-separation}). For SAT, an
observable rule resets on leaving a Hamming ball around the initial
assignment; it needs a distance bound but not the satisfying
assignment, and replay removes the supplied bound
(Theorem~\ref{thm:promise-hamming-restarts} and
Appendix~\ref{app:radius-free-replay}).

\paragraph{Related work.}
Computation with weak randomness
\citep{chor-goldreich-1988,doron-et-al-2023,doron-et-al-2025-online}
asks how to extract near-uniform bits from a defective source, or to
condense it into a shorter string with more randomness per symbol, and
then to run a randomized algorithm on the result, or to simulate the
algorithm by running it on many strings derived from the source
\citep{vazirani-vazirani-1985,zuckerman-1996-simulating,andreev-et-al-1999};
we instead take the algorithm as given, let it read the defective
source directly, once, and determine exactly how it fares. Extracting
first is not the question here: the model supplies no independent
seed, and for $s<1$ the class contains the semi-random sources of
Santha and Vazirani \citep{santha-vazirani-1986}, from which no
deterministic procedure produces a less biased bit.
Variable-length prefixes appear as a proof device in the
entropy-compression analysis of Achlioptas, Iliopoulos, and Vlassis
\citep{achlioptas-iliopoulos-vlassis-2017}.
The settings are inherited, and what is new is their analysis under
one correlated tape: the scheduling setting is classical
single-machine stochastic scheduling 
\citep{pinedo-2022}; the SAT walk is that of Austrin, Guruswami, and
H{\aa}stad \citep{austrin-guruswami-hastad-2017}; verified backtracking
is that of Feige, Raghavan, Peleg, and Upfal
\citep{feige-et-al-1994-noisy}; and restart schedules for independent
runs go back to Luby, Sinclair, and Zuckerman
\citep{luby-sinclair-zuckerman-1993}.

\paragraph{Organization.}
Section~\ref{sec:prefix-cover-robustness} gives the model and the
capacity identity; Sections~\ref{sec:repair-order}--\ref{sec:noisy-search} treat
the three settings and Section~\ref{sec:prefix-restarts} the restart
decision; Section~\ref{sec:related-work} asks further questions. Related work is detailed in
Appendix~\ref{app:related-work}, and all proofs are in the appendices.

\section{Prefix covers and source robustness}
\label{sec:prefix-cover-robustness}
\paragraph{Tapes and sources.}
Fix a finite alphabet $\mathcal A$ with at least two symbols and a
full-support distribution $\pi$ on it, and let $P=\pi^{\mathbb N}$ be
the product law. A \emph{tape} is a sequence in $\mathcal A^{\mathbb N}$,
and a \emph{word} is a finite sequence $w\in\mathcal A^*$; $w$ is a
\emph{prefix} of a tape $\omega$ if $\omega$ begins with $w$. The
\emph{cylinder} $[w]$ is the set of tapes with prefix $w$, and for a
law $\nu$ on tapes, $\nu[w]=\nu([w])$ is the probability that $w$ is a
prefix of the tape; under the reference law, $P[w]=\prod_i\pi(w_i)$.
The symbols are the actual primitive choices consumed by
the algorithm, and their alphabet is fixed once and for all: grouping
or recoding symbols changes the source assumption.

For $0<s\le1$ and $D\ge1$, define the source class
\begin{equation}\label{eq:all-prefix-source-class}
 \mathcal Q_{s,D}
 =\{\nu:\nu[w]\le DP[w]^s\text{ for every word }w\}.
\end{equation}
Here $\nu$ ranges over probability laws on whole tapes, and one
constant $D$ applies to prefixes of all lengths; we refer to
\eqref{eq:all-prefix-source-class} as the \emph{prefix probability
bound}, or the prefix bound for short. No randomness is reserved for the
algorithm outside this tape. Larger $s$ strengthens the guarantee,
whereas larger $D$ weakens it. At $s=D=1$ the class contains only $P$.
Equivalently, \eqref{eq:all-prefix-source-class} is a cumulative
information constraint: along every
tape whose finite prefixes have positive $\nu$-probability, the
increments $\log\nu[wa]-\log\nu[w]-s\log\pi(a)$ of the
log-likelihood ratio against the reference raised to the power $s$ sum
to at most $\log D$. The increments may be negative, so randomness left in the
tape early is banked and can be drawn on later.

The class is strictly larger than the class of sources with a
conditional guarantee. For a uniform $q$-symbol reference, a bound of
$q^{-s}$ on the conditional probability of every next symbol implies
membership in $\mathcal Q_{s,1}$, but the ternary stream of the introduction belongs
to $\mathcal Q_{18/19,1}$ and has conditional probability one for $m$
consecutive symbols. In the other direction, conditioning on a prefix
$w$ of positive $\nu$-probability gives $\nu[v\mid w]\le D_wP[v]^s$
with $D_w=DP[w]^s/\nu[w]$, which can be much larger than $D$. There is
no uniform reset of the guarantee after a history; this is why restarts
supply no fresh randomness. Appendix~\ref{app:prefix-cover-proof}
records further examples and explains why generators with a finite seed
lie outside the model, as they do outside every class with a positive
entropy rate.

\paragraph{Selectors, termination, and covers.}
A \emph{selector} is the algorithm's decision rule, a deterministic
policy in the language of control. It is \emph{nonanticipating}: at
each step it may use the prefix of the tape read so far, which
determines the entire history of states and actions, but no future
symbols. A prefix is \emph{live} if the algorithm has not terminated
after reading it. Let $\tau$ be the number of updates the algorithm makes,
including the update on which it terminates, so that $\{\tau>T\}$ is
the event that the algorithm is still live after $T$ updates. Write
$\Ncal=\{\tau=\infty\}$ for the nontermination event.

For a closed tape event $E$ (closed in the product topology; every
event determined by finitely many symbols is closed, and so is $\Ncal$),
define its \emph{prefix capacity}
\begin{equation}\label{eq:prefix-cover-content}
 \mathcal C_s(E)=\inf_{\mathcal B}\sum_{w\in\mathcal B}P[w]^s,
 \qquad E\subseteq\bigcup_{w\in\mathcal B}[w],
\end{equation}
where $\mathcal B$ ranges over finite prefix-free covers. Prefix-free
means that no chosen word extends another, so the cylinders are
disjoint; their lengths may differ. The empty word alone is a cover of
cost one. Every cover gives an upper bound on $\nu(E)$ for every
$\nu\in\mathcal Q_{s,D}$, by summing the source bounds over its words.
The identity below gives the converse: some admissible source attains
the best such bound.

\begin{theorem}[Classical tree capacity identity]
\label{thm:prefix-cover-robustness}
For every closed $E$, $0<s\le1$, and $1\le D<\infty$,
\begin{equation}\label{eq:exact-prefix-cover-robustness}
 \max_{\nu\in\mathcal Q_{s,D}}\nu(E)
 =\min\{1,D\mathcal C_s(E)\}.
\end{equation}
If $E$ depends on the first $T$ symbols, a minimizing cover uses words
of length at most $T$. For events $E_T$ that depend on finitely many
symbols and decrease to $\Ncal$,
\begin{equation}\label{eq:prefix-cover-infinite-limit}
 \max_{\nu\in\mathcal Q_{s,D}}\nu(\Ncal)
 =\lim_{T\to\infty}\max_{\nu\in\mathcal Q_{s,D}}\nu(E_T).
\end{equation}
\end{theorem}
A source is a flow of probability mass down the prefix tree, in which
the node $w$ has capacity $DP[w]^s$, and a cover is a cut separating
the root from every branch in $E$. The identity is the tree
max-flow/min-cut theorem, a minimax duality between sources and
covers, as in the tree proof of Frostman's lemma
\citep[Section~3.1]{bishop-peres-2017}. The one additional step,
completing a maximal flow into $E$ to a probability law on all tapes
that still respects every capacity, is carried out in
Appendix~\ref{app:prefix-cover-proof}. In the metric
$d_P(\omega,\omega')=P[\omega\wedge\omega']$ on tapes, where
$\omega\wedge\omega'$ is the longest common prefix, $\mathcal C_s(E)$
is the $s$-dimensional Hausdorff content of $E$, and the Hausdorff
dimension of the nontermination set is the threshold in $s$ below
which some finite-$D$ source is supported on it
(Corollary~\ref{cor:dimension-source-threshold}); no later result uses
this.

\begin{wrapfigure}{r}{0.42\textwidth}
\vspace{-1.2\baselineskip}
\centering
\begin{tikzpicture}[x=1.0cm,y=0.62cm,
  nd/.style={circle,draw,fill=white,inner sep=1.2pt,minimum size=3.5pt},
  cut/.style={circle,draw,fill=black,inner sep=1.2pt,minimum size=3.5pt},
  lab/.style={font=\footnotesize,inner sep=1pt},
  sub/.style={fill=black!12,draw=black!45,thin},
  nosub/.style={draw=black!45,thin,dotted}]
 \fill[sub] (-2.2,-1) -- (-3.0,-2.7) -- (-1.4,-2.7) -- cycle;
 \fill[sub] (-0.6,-4) -- (-1.1,-5.0) -- (-0.1,-5.0) -- cycle;
 \draw[nosub] (2.4,-2) -- (2.0,-2.9) -- (2.8,-2.9) -- cycle;
 \draw[nosub] (1.5,-3) -- (1.1,-3.9) -- (1.9,-3.9) -- cycle;
 \draw[nosub] (0.8,-4) -- (0.4,-4.9) -- (1.2,-4.9) -- cycle;
 \draw (0,0) -- (-2.2,-1) node[lab,midway,above left=-1.5pt] {$0$};
 \draw (0,0) -- (1.6,-1) node[lab,midway,above right=-1.5pt] {$1$};
 \draw (1.6,-1) -- (0.8,-2) node[lab,midway,above left=-1.5pt] {$0$};
 \draw (1.6,-1) -- (2.4,-2) node[lab,midway,above right=-1.5pt] {$1$};
 \draw (0.8,-2) -- (0.1,-3) node[lab,midway,above left=-1.5pt] {$0$};
 \draw (0.8,-2) -- (1.5,-3) node[lab,midway,above right=-1.5pt] {$1$};
 \draw (0.1,-3) -- (-0.6,-4) node[lab,midway,above left=-1.5pt] {$0$};
 \draw (0.1,-3) -- (0.8,-4) node[lab,midway,above right=-1.5pt] {$1$};
 \node[nd] at (0,0) {};
 \node[cut] (a) at (-2.2,-1) {};
 \node[nd] at (1.6,-1) {};
 \node[nd] at (0.8,-2) {};
 \node[nd] at (2.4,-2) {};
 \node[nd] at (0.1,-3) {};
 \node[nd] at (1.5,-3) {};
 \node[cut] (b) at (-0.6,-4) {};
 \node[nd] at (0.8,-4) {};
 \node[lab,anchor=east] at (-2.35,-0.95) {$2^{-1/2}$};
 \node[lab,anchor=east] at (-1.25,-4.35) {$2^{-2}$};
 \node[lab] at (-2.2,-2.3) {$[0]$};
 \node[lab] at (-0.6,-5.3) {$[1000]$};
 \draw[dashed,thick] plot[smooth,tension=0.7] coordinates
   {(-3.3,-1.75) (-2.2,-1) (-1.3,-1.6) (-1.1,-2.9) (-0.6,-4) (0.15,-4.95)};
 \node[lab,anchor=west] at (0.1,-5.15) {cut};
\end{tikzpicture}
\caption{Fair bits, $s=1/2$, $E=[0]\cup[1000]$ (shaded). The cover
$\{0,1000\}$ (filled nodes) is a cut separating the root from $E$, of
cost $2^{-1/2}+2^{-2}<1$; every single-depth cover costs at least one.}
\label{fig:cover-cut}
\vspace{-0.8\baselineskip}
\end{wrapfigure}
\paragraph{The cover recursion.}
Fix a deadline $T$, one symbol per update, and the event
$E=\{\tau>T\}$. For a prefix $w$,
let $V(w)=\mathcal C_s(E\cap[w])/P[w]^s$ be the cost of covering the
part of $E$ below $w$, measured in units of the cost $P[w]^s$ of
covering $[w]$ outright; thus $V(w)\le1$, $V(w)=1$ for a live word of
length $T$, and $V(w)=0$ for a terminated word. Working backwards, at a
prefix $w$ one either covers everything below it, at normalized cost
one, or covers its children separately, so
\begin{equation}\label{eq:reader-cover-recursion}
 V(w)=\min\left\{1,\sum_{a\in\mathcal A}\pi(a)^s V(wa)\right\},
\end{equation}
and the root value is $\mathcal C_s(E)$. Different branches can make different
choices, whereas a single-depth sum forces every branch to stop at the
same time. When the algorithm also chooses an update at each step, the
recursion takes a minimum over updates as well, and it becomes a
Bellman equation whose optimal policy is the robust algorithm
(Proposition~\ref{prop:controlled-prefix-value}).

Figure~\ref{fig:cover-cut} gives a minimal example: with fair bits
and $s=1/2$, the cover $\{0,1000\}$ of $E=[0]\cup[1000]$ costs
$2^{-1/2}+2^{-2}<1$, whereas depths $0,\dots,4$ cost
$1,\sqrt2,3/2,5/\sqrt8,9/4$ and refining further only increases the
cost: even the best single depth loses information.

\paragraph{What a single-depth bound retains.}
A single-depth cover uses the words of one length or of one update
count. With $\mathcal L_T$ the minimal prefixes completing $T$ updates,
including a terminating one, its cost is
$Z_T(s)=\sum_{w\in\mathcal L_T}P[w]^s$. Then $Z_T(1)=P\{\tau\ge T\}$ and
$\Prb_\nu\{\tau\ge T\}\le DZ_T(s)$; with one symbol per update,
$\widetilde Z_T(s)$ is the same sum over the length-$T$ words still
live, which bounds $\{\tau>T\}$. We call $Z_T(\theta)$ and
$\widetilde Z_T(\theta)$, for any exponent $\theta>0$, the \emph{prefix
power sums} at depth $T$. For a uniform $q$-symbol source,
$Z_T(s)=q^{(1-s)T}P\{\tau\ge T\}$, the independent-source tail inflated
by $q^{(1-s)T}$. A prefix cover can combine levels, and therefore uses
how the surviving prefixes extend one another; Section~\ref{sec:repair-order}
shows that this is what the power sums, like the running-time law,
cannot see.

\section{Scheduling retry tasks: an optimal order the stopping-time law cannot see}
\label{sec:repair-order}
The model is deliberately simple: attempting one task does not change
the others, so under independent
randomness the order is irrelevant, whereas under the prefix bound it
determines whether a linear deadline is missed with
probability one or with exponentially small probability.

\paragraph{The model.}
There are $n$ unfinished tasks and a finite symbol alphabet
$\mathcal A$. Task $i$ has a nonempty success set
$B_i\subseteq\mathcal A$. An attempt chooses an unfinished task and
reads the next symbol. The task completes if the symbol lies in $B_i$;
otherwise it remains unfinished. Completed tasks are never selected
again. The success sets may overlap. A deterministic nonanticipating selector may use the
entire observed history. Let $\tau$ count attempts until all tasks
complete. For $P=\pi^{\mathbb N}$ with full-support $\pi$, put
\[
 W=\sum_{a\in\mathcal A}\pi(a)^s,
 \qquad w_i=\sum_{a\in B_i}\pi(a)^s.
\]

\begin{theorem}[Optimal priority under a prefix probability bound]
\label{thm:retry-priority}\label{thm:retry-priority-research}
For every $0<s\le1$, always attempting an unfinished task with smallest
$w_i$ minimizes
$\max_{\nu\in\mathcal Q_{s,D}}\Prb_\nu\{\tau>T\}$
over all deterministic nonanticipating selectors, simultaneously for
every $T\ge0$ and $D\ge1$.
If $w_1\le\cdots\le w_n$, set $v_t(n+1)=0$, $v_0(j)=1$ for
$1\le j\le n$, and
\[
 v_{t+1}(j)=\min\{1,(W-w_j)v_t(j)+w_jv_t(j+1)\}.
\]
The exact optimum is $\min\{1,Dv_T(1)\}$. Once the weights are computed
and sorted, this takes $O(nT)$ arithmetic operations and $O(n)$ memory.
Every selector has the same stopping-time law under $P$ and the same
prefix power sums $Z_T(\theta)$ and $\widetilde Z_T(\theta)$ at every
depth $T$ and every exponent $\theta>0$.
\end{theorem}

For a uniform reference source, the rule is to finish a task with the
fewest successful symbols first. It is an index policy, each task
having the fixed index $w_i$ and the smallest index being served first,
in the tradition of priority policies for stochastic scheduling
\citep{mohring-schulz-uetz-1999,pinedo-2022}. For a nonuniform
reference, the weight $w_i$ of the success set replaces this count. It
is a cover weight rather than a probability: all symbol weights sum to
$W$, which may exceed one. With singleton success sets,
the rule still orders tasks by increasing ordinary success probability;
in general the order can depend on $s$.

The recursion in the theorem is the cover recursion
\eqref{eq:reader-cover-recursion} on the state $j$, the index of the
first unfinished task in sorted order: an attempt at task $j$ fails at
cost $W-w_j$ times the value with the same tasks remaining, or succeeds
at cost $w_j$ times the value with one fewer task. The content of the theorem is that nothing
is gained by attempting tasks in any other order, or by switching
orders in response to the observed history. Appendix~\ref{app:retry-priority-proof}
proves this by a weighted exchange inequality that is compatible with
the truncation at one; the exchange inequality is the step that reduces
a dynamic program over all $2^n$ subsets of unfinished tasks to the $n$
suffixes of the sorted list.

\paragraph{A quantitative separation.}
Take $n$ tasks of type $A$, successful on $0$, and one task of type $B$,
successful on $1$. Write $\Lambda_A$ for the order that completes all
$A$ tasks first and $\Lambda_B$ for the order that completes $B$ first.

\begin{theorem}[Equal prefix sums, different robustness]
\label{thm:repair-order-multiscale}
Fix $P(0)=15/16$, $P(1)=1/16$, and
$\mathcal Q=\mathcal Q_{1/2,1}$. For every $n\ge1$,
\begin{align}
 \max_{\nu\in\mathcal Q}\Prb_\nu^{\Lambda_A}\{\tau>8n\}&=1,
 \label{eq:repair-order-bad}\\
 \left(\frac{15}{16}\right)^{4n}
 \le\max_{\nu\in\mathcal Q}\Prb_\nu^{\Lambda_B}\{\tau>8n\}
 &\le\frac98\left(\frac{225}{256}\right)^n<1.
 \label{eq:repair-order-good}
\end{align}
One source $\nu_n\in\mathcal Q$ gives survival probability one for
$\Lambda_A$ and at most the upper bound in \eqref{eq:repair-order-good}
for $\Lambda_B$. Moreover, the lower bound in
\eqref{eq:repair-order-good} holds for every selector, $\Lambda_B$ is
minimax optimal at every deadline, and all deterministic
nonanticipating selectors have identical stopping-time laws under $P$
and identical prefix power sums at every depth and every exponent
$\theta>0$.
\end{theorem}

Since $\Lambda_B$ attempts the task of smallest weight first,
Theorem~\ref{thm:retry-priority} gives its optimality and the equal
laws and power sums; the lower bound holds for every selector because
its proof uses a source whose first $8n$ symbols are all zeros with
probability $(15/16)^{4n}$.

\paragraph{Why all power sums agree.}
The prefix power sums with exponent $\theta$ are, up to the factor
$(\sum_a\pi(a)^\theta)^T$, survival probabilities under the tilted
independent law with symbol probabilities proportional to
$\pi(a)^\theta$. Under any independent law, preassign an independent
stream of symbols to each task: the next unread symbol of any stream
has the same law whatever the selector has observed, so the total
number of attempts is the sum of the tasks' independent geometric
workloads, whichever stream the selector reads next
(Appendix~\ref{app:retry-priority-proof}).

\paragraph{Why the robust tails differ.}
The bad order $\Lambda_A$ postpones the only task that needs a $1$. An
adversarial source exploits this by generating zeros with probability
$3/4$ until the $n$th zero appears, then forcing zeros through time
$8n$, and then following $P$. To check the source constraint, track the
ratio $\nu[w]/P[w]^{1/2}$ along the tape, which must stay at most one.
A random one leaves it unchanged, each random zero multiplies it by
$(3/4)/(15/16)^{1/2}<1$, and each forced zero multiplies it by
$(16/15)^{1/2}>1$; the probability $3/4$ is chosen so that the $n$
random zeros outweigh the at most $7n$ forced ones. Under $\Lambda_A$,
either some $A$ task is still unfinished at time $8n$, or the $A$ tasks
finish first and task $B$ then sees only zeros for the rest of the
deadline.

For the good order $\Lambda_B$, no single depth gives a useful bound:
$\Lambda_A$ survives with probability one, so every single-depth cover
of its survival event costs at least one, and by
Theorem~\ref{thm:retry-priority} the same sums are those of $\Lambda_B$.
A cover at two depths does better. Cover first by the prefix $0^{4n}$, on
which $B$ is still unfinished at time $4n$. Otherwise the first $1$
completes $B$ by time $4n$, and survival to time $8n$ requires fewer
than $n$ zeros among the remaining symbols; a generating-function bound
covers these paths at their full depth. The resulting mixed-depth cover
proves \eqref{eq:repair-order-good}, and
Appendix~\ref{app:repair-order-proof} checks the source constraint and
both bounds.

Thus the cover resolves an algorithmic choice that every single-depth
power sum leaves open. One tape shared by several independent
processes is likewise delicate:
Appendix~\ref{app:component-proof} gives processes, each robust alone,
that one admissible shared source keeps from jointly terminating under
every interleaving.

\section{Local SAT search: sharp bounds and restarts}
\label{sec:promise-sat}
A flip can satisfy one clause while violating several others, so the
number of violated clauses is not a reliable measure of progress for a
local walk. The promise below supplies a different measure: the Hamming
distance to an unknown assignment that satisfies a majority of each
clause's literals. The walk never sees this distance, but every flip
it makes changes it by exactly one.

\paragraph{The promise and the walk.}
Fix integers $k/2<r<k$. A $k$-CNF is \emph{strongly $r$-satisfiable}
if some assignment $x^*$ makes at least $r$ literals true in each
clause; clauses use $k$ distinct variables. We call such an $x^*$ a
strong witness. The walk need only find an ordinary satisfying
assignment, not a strong witness. At each step the algorithm chooses
any violated clause and flips the variable indexed, under a fixed
ordering of its literals, by the next source symbol, which is uniform
on $[k]$ under the reference law $P$; the clause choice may depend on
the complete past. This is an established
polynomial-time algorithm for the promise problem
\citep[Proposition~6.1]{austrin-guruswami-hastad-2017}: when a clause
is violated, all of its literals are false, so every literal that $x^*$
satisfies sits on a variable where the current assignment and $x^*$
disagree, and at least $r$ of the $k$ possible flips decrease the
distance to $x^*$. Put
\begin{equation}\label{eq:promise-parameters}
 \lambda=\sqrt{\frac{r}{k-r}}>1,
 \qquad s_* = \log_k\bigl(2\sqrt{r(k-r)}\bigr)\in(0,1).
\end{equation}

\begin{theorem}[Robustness from distance to a solution]
\label{thm:promise-robust}
Let $\Phi$ be a strongly $r$-satisfiable $k$-CNF on $n$ variables, and let $h$ be the
Hamming distance of its initial assignment from a strong witness $x^*$.
Put $L=h+\log D+1$ and $\delta=s-s_*$. There are constants $c,C>0$,
depending only on $k,r$, such that for every deterministic
nonanticipating clause selector, $s\in[s_*,1]$, $D\ge1$,
$\nu\in\mathcal Q_{s,D}$, and integer $T\ge0$,
\begin{align}
 \Prb_\nu\{\tau>T\}
 &\le \min\{1,CD\lambda^h e^{-cT^{1/3}},D\lambda^h k^{-\delta T}\},
 \label{eq:promise-robust}\\
 \E_\nu\tau&\le C\min\{L^3,L/\delta\}.
 \label{eq:promise-crossover}
\end{align}
We interpret $L/0=\infty$. Here $\tau$ counts variable flips until the
formula is satisfied; there is no restriction on clause overlap.
\end{theorem}

Under independent randomness the walk's expected
flip count is at most $kh/(2r-k)$, by the drift of the distance. The
robust bound depends on the distance to a strong witness and not on
the number of variables, because the proof pays only for excursions in
the distance. The statement counts flips; finding a violated clause
takes additional work.

\paragraph{Why the threshold is $s_*$.}
At a violated clause, at least $r$ of the $k$ labels (next symbols)
decrease the distance $d$ to $x^*$ by one and the rest increase it by
one; the distance process need not be Markov, and in a cover sum each
label carries the weight $k^{-s}$. To bound the cover recursion by a
function $\phi(d)$ of the distance alone, nondecreasing in $d$, it
suffices that $\phi(d)\ge a_s\phi(d-1)+b_s\phi(d+1)$ wherever the
recursion continues: extra decreasing labels only lower the children's
sum, so the inequality propagates backward, where
\[
 a_s=rk^{-s},\qquad b_s=(k-r)k^{-s}.
\]
These are weights in a cover sum, not conditional probabilities of the
actual next move; they sum to $k^{1-s}\ge1$, the cost of replacing a word
by its $k$ children in a cover. Dividing $\phi(d)$ by $\lambda^d$
turns this two-sided comparison into a symmetric one with total weight
$2\sqrt{a_sb_s}$ per step; the threshold $s_*$ is where this total
weight equals one. At that point $a_{s_*}=\lambda/2$ and
$b_{s_*}=1/(2\lambda)$, so after the change of scale the comparison is
a symmetric simple random walk, with its diffusive time scale. Above
the threshold the total weight is $k^{-\delta}<1$ per step, which gives
the exponential tail in \eqref{eq:promise-robust}.

\paragraph{Why different depths help.}
At the threshold, a single-depth cover decays only polynomially
(Theorem~\ref{thm:promise-sharp}). The stretched-exponential tail comes
from covering at two kinds of depth. Choose a barrier $m>h$.
Cover a surviving run when it first reaches distance $m$, or at time
$T$ if it has stayed below $m$ throughout. The first part costs at most
$(h/m)\lambda^{h-m}$, because the function $(d/m)\lambda^{d-m}$ is
harmonic for the two-sided comparison and equals one at the barrier;
so the source can give a distant excursion only exponentially small
probability in $m$. The second part
costs at most $C\lambda^h e^{-cT/m^2}$, because a symmetric walk
confined to a strip of width $m$ decays on time scale $m^2$. Balancing
$m$ against $T/m^2$ gives $m\asymp T^{1/3}$ and the stated tail. The
first-hit prefixes have different lengths even though there is only one
spatial barrier.

\paragraph{A reusable barrier principle.}
Nothing in this argument is specific to SAT: the distance is a
Lyapunov potential whose drift condition is replaced by a worst-case
count of favorable symbols.
Lemma~\ref{lem:promise-potential-barrier} in
Appendix~\ref{app:promise-proofs} states it for any deterministic
process driven by symbols in $[k]$: assign a positive integer potential
to every live history, with potential zero on termination, and assume
that at every live history at least $r$ of the $k$ next symbols either
terminate or decrease the potential by at least one, while every other
symbol either terminates or increases it by at most one. Then, for every
$\nu\in\mathcal Q_{s_*,D}$ and every barrier $m$ above the initial
potential $h$, the two-part bound above holds with
$\lambda=\sqrt{r/(k-r)}$: survival at time $T$ has probability at most
$\min\{1,D[(h/m)\lambda^{h-m}+C\lambda^he^{-cT/m^2}]\}$, and the
expected running time is $O(L^3)$ with $L=h+\log D+1$. The potential
need not be Markov, the state space need not be finite, and blocks of
symbols with nonuniform weights are allowed, as in
Section~\ref{sec:noisy-search}. For comparison, a single-depth argument
that reflects at the global distance bound $n$ gives only
$\E_\nu\tau=O_{k,r}(n^2(h+\log D+1))$ at $s=s_*$; for fixed $D$, the
barrier buys $h^2$ in place of $n^2$, and the critical tail. The barrier sits at an
unknown distance; an observable Hamming ball replaces it below.

\paragraph{The threshold is sharp.}
The complete formula family $\Phi_n$ of the theorem below, every
clause with $r$ positive and $k-r$ negative literals on distinct
variables, supplies exact lower bounds. A clause of $\Phi_n$ is
violated exactly when its $r$ positive variables are zero and its $k-r$
negative variables are one, so an assignment with
$d$ zeros is live exactly when $r\le d\le n-k+r$. The walk is therefore
a one-dimensional process on these $N-1=n-k+1$ live distances, absorbed
at both ends, and its transitions are the same for every clause
selector.

\begin{theorem}[Sharp threshold and a gap from every single depth]
\label{thm:promise-sharp}
Let $\Phi_n$ contain all clauses with $r$ positive and $k-r$ negative
literals on distinct variables, with positive literals ordered first,
and put $N=n-k+2\ge3$.
\begin{enumerate}[label=(\roman*),leftmargin=*]
\item For every $0<s<s_*$, some $n$, initial assignment, and single
source in $\mathcal Q_{s,1}$ force nontermination under every
nonanticipating clause selector.
\item Start with exactly $r$ zero variables, so that $h=r$. Uniformly in
$s\in[s_*,1]$, $N\ge3$, and $T\ge1$, the worst survival probability is
\begin{equation}\label{eq:promise-profile}
 A_{N,T}(s):=\max_{\nu\in\mathcal Q_{s,1}}\Prb_\nu\{\tau>T\}
 =\exp\!\left[-\Theta_{k,r}\!\left(T^{1/3}+(s-s_*)T+T/N^2\right)\right].
\end{equation}
At $s=s_*$, for the same event $E_T=\{\tau>T\}$, let $F_t(E_T)$ be the minimum
cost of a cover using only length-$t$ cylinders. Then
\begin{equation}\label{eq:promise-flat-gap}
 \inf_{t\ge0}F_t(E_T)=\Theta_{k,r}((T+1)^{-3/2})
 \qquad(1\le T\le N-2).
\end{equation}
\end{enumerate}
\end{theorem}

Thus $s_*$ is the sharp uniform threshold for this walk over the
promise class; the three terms in \eqref{eq:promise-profile} are the
critical cube-root tail, the gain from a margin above the threshold,
and the gain from the global distance bound $n$. At $s=s_*$ and the
deadline $T=n-k$, variable-depth covers give $e^{-\Theta(n^{1/3})}$,
while every single-depth cover of the same event costs at least order
$n^{-3/2}$. The symmetric family supplies the lower bounds only;
Theorem~\ref{thm:promise-robust} requires no symmetry. The profile is
for $D=1$; Theorem~\ref{thm:promise-distance-sharp} supplies the
dependence on $D$.

For $k=3,r=2$, $s_*=\log_3(2\sqrt2)\approx0.9464$, and $18$ is the
least number of fresh draws per repeated draw in the introductory
stream whose exponent, $18/19\approx0.9474$, exceeds it. The threshold
falls with the clause width: for $r=k-1$ it is $\log_k(2\sqrt{k-1})$,
which is $0.86$ at $k=5$ and tends to $1/2$ as $k\to\infty$. For
$k=3,r=2$ a strong witness can be found by
2-SAT, replacing each clause by its three pairwise subclauses, so the
result concerns the robustness of the specified walk, not tractability.
The obstruction below $s_*$ need not cycle:
Corollary~\ref{cor:promise-full-support} gives a full-support source
with failure probability at least $1-\varepsilon$ for every selector,
under which the nonterminating trajectory of each memoryless selector
is almost surely not eventually periodic.

\paragraph{The whole transition is sharp.}
Define
\[
 \Psi_s(L)=\min\{L^3,L/(s-s_*)\},\qquad L/0=\infty.
\]
\begin{theorem}[Matching delay across the source threshold]
\label{thm:promise-distance-sharp}
Fix $k/2<r<k$. There is a constant $c_{k,r}>0$ such that for every
$s\in[s_*,1]$, integer $h\ge r$, and $D\ge1$, some complete formula
$\Phi_n$ from Theorem~\ref{thm:promise-sharp} has an initial assignment
that does not satisfy $\Phi_n$, at distance $h$ from its unique strong
witness, and a single source
$\nu\in\mathcal Q_{s,D}$ for which, simultaneously for every clause
selector,
\[
 \tau\ge c_{k,r}\Psi_s(h+\log D+1)\qquad\nu\text{-almost surely}.
\]
The number of variables satisfies $n=O_{k,r}(\Psi_s(h+\log D+1))$.
In particular, \eqref{eq:promise-crossover} is sharp, uniformly in
initial distance, source slack, and $s\in[s_*,1]$.
\end{theorem}

\begin{wrapfigure}{r}{0.44\textwidth}
\vspace{-1.0\baselineskip}
\centering
\begin{tikzpicture}[lab/.style={font=\footnotesize,inner sep=1pt}]
 \draw[->] (0,0) -- (5.95,0) node[lab,below left=-1pt] {$j$};
 \draw[->] (0,0) -- (0,2.85) node[lab,above] {$C_t(j)$};
 \draw (0,2.60) -- (-0.08,2.60) node[lab,left] {$1$};
 \draw[black!25] (0,2.60) -- (5.60,2.60);
 \draw[solid,thick] plot coordinates {(0.00,0.00) (0.07,0.24) (0.14,0.56) (0.21,1.11) (0.28,1.55) (0.35,2.18) (0.42,2.46) (0.49,2.60) (0.56,2.60) (0.63,2.60) (0.70,2.60) (0.77,2.60) (0.84,2.60) (0.91,2.60) (0.98,2.60) (1.05,2.60) (1.12,2.60) (1.19,2.60) (1.26,2.60) (1.33,2.60) (1.40,2.60) (1.47,2.60) (1.54,2.60) (1.61,2.60) (1.68,2.60) (1.75,2.60) (1.82,2.60) (1.89,2.60) (1.96,2.60) (2.03,2.60) (2.10,2.60) (2.17,2.60) (2.24,2.60) (2.31,2.60) (2.38,2.60) (2.45,2.60) (2.52,2.60) (2.59,2.60) (2.66,2.60) (2.73,2.60) (2.80,2.60) (2.87,2.60) (2.94,2.60) (3.01,2.60) (3.08,2.60) (3.15,2.60) (3.22,2.60) (3.29,2.60) (3.36,2.60) (3.43,2.60) (3.50,2.60) (3.57,2.60) (3.64,2.60) (3.71,2.60) (3.78,2.60) (3.85,2.60) (3.92,2.60) (3.99,2.60) (4.06,2.60) (4.13,2.60) (4.20,2.60) (4.27,2.60) (4.34,2.60) (4.41,2.60) (4.48,2.60) (4.55,2.60) (4.62,2.60) (4.69,2.60) (4.76,2.60) (4.83,2.60) (4.90,2.60) (4.97,2.60) (5.04,2.60) (5.11,2.60) (5.18,2.60) (5.25,2.60) (5.32,2.60) (5.39,2.60) (5.46,2.60) (5.53,2.60) (5.60,2.60)};
 \draw[densely dashed,thick] plot coordinates {(0.00,0.00) (0.07,0.01) (0.14,0.02) (0.21,0.05) (0.28,0.09) (0.35,0.15) (0.42,0.23) (0.49,0.35) (0.56,0.51) (0.63,0.72) (0.70,0.97) (0.77,1.30) (0.84,1.63) (0.91,2.02) (0.98,2.35) (1.05,2.60) (1.12,2.60) (1.19,2.60) (1.26,2.60) (1.33,2.60) (1.40,2.60) (1.47,2.60) (1.54,2.60) (1.61,2.60) (1.68,2.60) (1.75,2.60) (1.82,2.60) (1.89,2.60) (1.96,2.60) (2.03,2.60) (2.10,2.60) (2.17,2.60) (2.24,2.60) (2.31,2.60) (2.38,2.60) (2.45,2.60) (2.52,2.60) (2.59,2.60) (2.66,2.60) (2.73,2.60) (2.80,2.60) (2.87,2.60) (2.94,2.60) (3.01,2.60) (3.08,2.60) (3.15,2.60) (3.22,2.60) (3.29,2.60) (3.36,2.60) (3.43,2.60) (3.50,2.60) (3.57,2.60) (3.64,2.60) (3.71,2.60) (3.78,2.60) (3.85,2.60) (3.92,2.60) (3.99,2.60) (4.06,2.60) (4.13,2.60) (4.20,2.60) (4.27,2.60) (4.34,2.60) (4.41,2.60) (4.48,2.60) (4.55,2.60) (4.62,2.60) (4.69,2.60) (4.76,2.60) (4.83,2.60) (4.90,2.60) (4.97,2.60) (5.04,2.60) (5.11,2.60) (5.18,2.60) (5.25,2.60) (5.32,2.60) (5.39,2.60) (5.46,2.60) (5.53,2.60) (5.60,2.60)};
 \draw[dash dot,thick] plot coordinates {(0.00,0.00) (0.91,0.01) (0.98,0.01) (1.05,0.02) (1.12,0.02) (1.19,0.03) (1.26,0.04) (1.33,0.06) (1.40,0.09) (1.47,0.12) (1.54,0.17) (1.61,0.23) (1.68,0.31) (1.75,0.41) (1.82,0.54) (1.89,0.71) (1.96,0.92) (2.03,1.18) (2.10,1.47) (2.17,1.80) (2.24,2.13) (2.31,2.43) (2.38,2.60) (2.45,2.60) (2.52,2.60) (2.59,2.60) (2.66,2.60) (2.73,2.60) (2.80,2.60) (2.87,2.60) (2.94,2.60) (3.01,2.60) (3.08,2.60) (3.15,2.60) (3.22,2.60) (3.29,2.60) (3.36,2.60) (3.43,2.60) (3.50,2.60) (3.57,2.60) (3.64,2.60) (3.71,2.60) (3.78,2.60) (3.85,2.60) (3.92,2.60) (3.99,2.60) (4.06,2.60) (4.13,2.60) (4.20,2.60) (4.27,2.60) (4.34,2.60) (4.41,2.60) (4.48,2.60) (4.55,2.60) (4.62,2.60) (4.69,2.60) (4.76,2.60) (4.83,2.60) (4.90,2.60) (4.97,2.60) (5.04,2.60) (5.11,2.60) (5.18,2.60) (5.25,2.60) (5.32,2.60) (5.39,2.60) (5.46,2.60) (5.53,2.60) (5.60,2.60)};
 \draw[densely dotted,thick] plot coordinates {(0.00,0.00) (3.57,0.01) (3.64,0.01) (3.71,0.01) (3.78,0.02) (3.85,0.02) (3.92,0.03) (3.99,0.04) (4.06,0.05) (4.13,0.07) (4.20,0.10) (4.27,0.13) (4.34,0.17) (4.41,0.23) (4.48,0.30) (4.55,0.40) (4.62,0.52) (4.69,0.68) (4.76,0.87) (4.83,1.11) (4.90,1.38) (4.97,1.70) (5.04,2.03) (5.11,2.35) (5.18,2.58) (5.25,2.60) (5.32,2.60) (5.39,2.60) (5.46,2.60) (5.53,2.60) (5.60,2.60)};
 \draw (0.49,0.06) -- (0.49,-0.06) node[lab,below] {$7$};
 \draw (1.12,0.06) -- (1.12,-0.06) node[lab,below] {$16$};
 \draw (2.38,0.06) -- (2.38,-0.06) node[lab,below] {$34$};
 \draw (5.25,0.06) -- (5.25,-0.06) node[lab,below] {$75$};
 \node[lab,anchor=south west] at (0.55,2.63) {$t=10$};
 \node[lab,anchor=west] at (0.99,1.69) {$t=10^2$};
 \node[lab,anchor=west] at (2.11,1.17) {$t=10^3$};
 \node[lab,anchor=east] at (4.61,0.65) {$t=10^4$};
\end{tikzpicture}
\caption{The capped recursion of Lemma~\ref{lem:promise-moving-plateau}
at criticality ($k=3$, $r=2$; $\lambda=\sqrt2$, $\epsilon=0$). The value
is one beyond a frontier at $j=7,16,34,75$ for $t=10,\dots,10^4$, which
advances like $t^{1/3}$; while $j=h-r+1$ lies in this plateau, the
worst-case survival probability is one even with $D=1$.}
\label{fig:frontier}
\vspace{-0.6\baselineskip}
\end{wrapfigure}
One source forces the entire delay, for every selector: at the
threshold and fixed $D$ it has order $h^3$, against an
independent-source expectation of $O_{k,r}(h)$. With $\delta=s-s_*>0$
the crossover is at $\delta L^2\asymp1$, the cubic bound governing
smaller $L$ and $L/\delta$ larger $L$.

\paragraph{Why the delay is forced with probability one.}
By Theorem~\ref{thm:prefix-cover-robustness}, the worst-case
probability of $\{\tau>T\}$ is one exactly when this event has capacity
at least $1/D$, and tree max-flow then supplies a source supported
entirely on executions that survive time $T$. For the complete
formulas, the capacity is the cover recursion on the distance line.
For $T=0$ every live distance has value one. As $T$ grows, the values
first drop near the end adjacent to the witness, and the frontier of
the region where
the value is still exactly one moves away from that end at speed at
most of order $\epsilon+m^{-2}$ when it is at distance $m$, where
$\epsilon=(s-s_*)\log k$ (Lemma~\ref{lem:promise-moving-plateau}). The
initial distance stays inside the value-one region until the frontier
reaches it (Figure~\ref{fig:frontier}), which takes cubic time at
criticality and linear time at a fixed margin; this handles $D=1$. The cube-root-to-linear transition of
such a frontier is known in branching processes
\citep{mallein-2019-consistent,berestycki-et-al-2024-transition};
Appendix~\ref{app:promise-proofs} proves the comparison and the separate
cover lower bound that supplies the dependence on $D$. The theorem takes
$n$ of order $\Psi_s(L)$ so that the far end of the live interval is
out of reach within the horizon. This is a convenience: the far end
absorbs too, but moving away from it is the walk's favorable direction,
so it erodes the plateau by a bounded number of states only.

\begin{corollary}[Matching delay at linear distance]
\label{cor:promise-linear-distance}
There are constants $c,\delta_0>0$ and $h_0\ge1$, depending only on
$k,r$, such that for every $s\in[s_*,s_*+\delta_0]$, every $n$, and
every integer $h$ with $h_0\le h\le n-h_0$, the formula $\Phi_n$ with
an initial assignment at distance $h$ from its strong witness admits a
single source $\nu\in\mathcal Q_{s,1}$ for which, simultaneously for
every clause selector,
\[
 \tau>c\min\{h^3,h/(s-s_*)\}\qquad\nu\text{-almost surely},
\]
with $h/0=\infty$. In particular, for $s=s_*$, $h=\lfloor n/2\rfloor$,
$n\ge2h_0$, and fixed $D$, the worst-case expected number of flips over
$\mathcal Q_{s_*,D}$ is $\Theta_{k,r,D}(n^3)$ for every selector,
against $\Theta_{k,r}(n)$ under independent uniform labels.
\end{corollary}

\paragraph{An observable restart rule.}
The cover proof identifies excursions far from a strong witness as
the obstruction to a better tail; although the witness is unknown, a
radius bound cuts off these excursions using only the current and
initial assignments.

\begin{theorem}[Restarting on leaving a Hamming ball]
\label{thm:promise-hamming-restarts}
Suppose the formula has a strong witness $x^*$
at Hamming distance $h\ge1$ from the initial assignment $x_0$.
Let $H\ge h$ be a supplied positive integer.
Run the violated-clause walk, and reset to $x_0$ whenever an
unsatisfying assignment leaves the Hamming ball of radius $2H$
around $x_0$. Continue with the next unused source symbol.
The rule uses $H$ and $x_0$, but does not use $x^*$.
There are constants $c,C>0$, depending only on $k,r$, such that
for every clause selector, $D\ge1$, $s\in[s_*,1]$, and
$\nu\in\mathcal Q_{s,D}$, its stopping time $\tau_H$ satisfies
\begin{equation}\label{eq:promise-hamming-restart-tail}
 \Prb_\nu\{\tau_H>T\}
 \le \min\{1,CD\lambda^h e^{-cT/H^2}\},
 \qquad T\in\mathbb Z_{\ge0}.
\end{equation}
In particular,
\[
 \E_\nu\tau_H=O_{k,r}\!\left(H^2(h+\log D+1)\right).
\]
Here $\tau_H$ counts source-driven flips; resets at most double the
number of variable updates, the distance from $x_0$ costs constant work
per flip, and finding a violated clause is separate work.
\end{theorem}

At $s=s_*$, a failure probability at most $\varepsilon$ costs
$O_{k,r}(H^2(h+\log D+\log(1/\varepsilon)+1))$ flips, against
$O_{k,r}((h+\log D+\log(1/\varepsilon)+1)^3)$ for the original walk:
for $H\asymp h$ and fixed $D$ both are cubic in $h$, and the gain is
that the confidence enters linearly rather than cubically. The
algorithm needs neither the witness nor its distance $h$, only a valid
upper bound $H$.

The rule works because an assignment $y$ just outside the ball has
$\operatorname{dist}(y,x^*)\ge2H+1-h>h$ by the triangle inequality, so
resetting to $x_0$ strictly decreases the unknown witness distance; the
restarted walk keeps that distance at most $h+2H$ and retains the
favorable-flip comparison, and a single-depth estimate on this bounded
range, through the largest eigenvalue of the weighted transitions,
gives the exponential tail (Appendix~\ref{app:promise-proofs}).

The supplied radius can be removed by replaying the archived tape in
copies with geometrically increasing radii, at the price of
$O(\log(n+2))$ search states and tape storage
(Theorem~\ref{thm:promise-radius-free-replay}). Each copy reads the
tape from its beginning, so the prefix bound applies to it with the
original $D$, whereas a single walk that enlarged its radius after
repeated resets would run its decisive phase on the suffix left by the
earlier phases, with the inflated constant $D_w$ of
Section~\ref{sec:prefix-cover-robustness};
Appendix~\ref{app:radius-free-replay} shows that this inflation is
attained, and whether some single walk can match the replay guarantee
without storing the tape is open.

\section{Noisy predecessor search under a prefix probability bound}
\label{sec:noisy-search}
Noisy search must both finish and return the right answer. Under
independent errors, repeated checks provide confidence; under our
source condition, consecutive errors may be predictable from the past,
so the checks must be analyzed through the whole error sequence rather
than one check at a time. We use verified backtracking and verification
chains, which are classical \citep[Section~3]{feige-et-al-1994-noisy}.
Our chain lengths depend on leaf depth, giving a guarantee in terms of
the correct leaf's depth without knowing it or the tree height; sharper
bounds under independent noise are known
\citep{karp-kleinberg-2007,ben-or-hassidim-2008,epa-gan-wirth-2019-result-sensitive,dereniowski-et-al-2025-noisy}.

Let $x_1<\cdots<x_n$ ($n\ge 1$) be sorted keys and let $x\in\mathbb R$ be the
target. The task is to return the index of the largest key not exceeding $x$,
the \emph{predecessor} of $x$. A primitive query asks whether $x<x_i$. Write $e_t=1$
when the answer to the $t$th primitive query is flipped and $e_t=0$
otherwise. Fix $1/2<p<1$ and
assume, for every prefix $e$ of the error sequence $(e_t)$,
\begin{equation}\label{eq:noisy-search-prefix}
 \nu[e]\le D\bigl(p^{|e|-|e|_1}(1-p)^{|e|_1}\bigr)^s.
\end{equation}
Thus the reference law $P$ has independent errors of probability $1-p$;
the actual law need not give any conditional correctness guarantee.

\paragraph{The walk on the tree.}
Fix an ordered binary search tree whose leaves represent the
predecessor intervals, so that an internal node represents an interval
$[x_i,x_j)$ and splits it at some key. The algorithm walks on this
tree. At each node it makes three logical comparisons. Two
\emph{verification} comparisons test whether $x$ lies in the node's
interval, by asking whether $x<x_i$ and whether $x<x_j$; one
\emph{split} comparison chooses the child. If verification fails, the
walk moves to the parent; otherwise it descends to the chosen child.
This is verified backtracking \citep{feige-et-al-1994-noisy}.

Each original leaf is extended by a chain of nodes that all represent
the leaf's interval, and a chain node's parent is the previous node of
the chain. At a chain node the split comparison is replaced by a dummy
comparison, a primitive query that is made but whose answer is ignored;
verification continues, failed verification moves back one node, and
success advances along the chain. The walk halts at the first chain
endpoint it reaches and returns that leaf's label (comparisons against
an infinite boundary are likewise dummies). Every
logical comparison is answered by the majority of $\ell$ repetitions
of the same primitive query, $\ell$ odd and fixed below, and all $3\ell$
primitive queries at a node are made before the move, whether or not
their answers are needed; so every move consumes exactly $3\ell$ raw
error symbols, and the prefix bound is applied to prefixes that end at
block boundaries.

The chains serve the error guarantee. Measure distance from the
endpoint of the correct leaf's chain: a move with all three majorities
correct brings the walk one step closer to it, by correctly rejecting
a node whose interval excludes $x$ or correctly descending or advancing
from one whose interval contains it, and any other move changes the
distance by at most one. Longer chains at deeper leaves place every
wrong endpoint farther from the correct endpoint than the walk's
starting point, by a margin that controls all incorrect outputs at
once.

\paragraph{The source threshold.}
Choose the least odd $\ell=\ell(p)$ for which the majority of $\ell$
independent reference answers is correct with probability greater than
$2^{-1/3}$, so that three independent majorities are all correct with
probability greater than $1/2$. Let $G\subseteq\{0,1\}^{3\ell}$ contain
the raw error blocks in which all three majorities are correct, and put
\[
 a_s=\sum_{e\in G}P[e]^s,\qquad
 b_s=\sum_{e\in\{0,1\}^{3\ell}\setminus G}P[e]^s,
 \qquad 2\sqrt{a_{s_c}b_{s_c}}=1,
 \qquad \lambda=\sqrt{a_{s_c}/b_{s_c}}.
\]
These play the roles of $a_s$, $b_s$, and $\lambda$ in Section~\ref{sec:promise-sat},
with a block of $3\ell$ error symbols in place of one label. The
solution $s_c\in(0,1)$ exists and is unique because $a_sb_s$ is
strictly decreasing in $s$, and $\lambda>1$ because $a_{s_c}>1/2$
(Appendix~\ref{app:noisy-predecessor}). Blocks in $G$ are the
favorable moves above; every other block is treated as potentially
adverse, with no independence assumed among the majority answers, so
$s_c$ is a sufficient threshold rather than a sharp one. The favorable
set in fact depends on the state, and $G$ is its worst case: at a
chain node the split comparison is a dummy, and at a node whose
interval excludes $x$ one correct verification majority already moves
the walk toward the correct endpoint. An excursion beyond the starting
distance must enter such nodes, so no instance realizes the adverse
weight $b_s$ along a long excursion, and the sharp threshold may lie
below $s_c$ (Section~\ref{sec:related-work}).

\begin{theorem}[Verified search with correlated errors]
\label{thm:noisy-predecessor}
Take any ordered binary search tree whose leaves represent the
predecessor intervals, and let $h$ be the depth of the correct leaf.
Put $L_\eta=\max\{1,\lceil\log_\lambda(D/\eta)\rceil\}$ and attach a
chain of $\operatorname{depth}(u)+L_\eta$ edges to each original leaf $u$.
For every $D\ge1$, $0<\eta<1$, $s\in[s_c,1]$, and law satisfying
\eqref{eq:noisy-search-prefix}, the walk above halts almost
surely and returns the predecessor of $x$ with error probability at
most $\eta$. The number $Q$ of primitive queries satisfies
\[
 \E_\nu Q\le C_p\min\left\{K^3,\frac{K}{s-s_c}\right\},
 \qquad K=h+\log D+\log(1/\eta)+1,
\]
where $K/0=\infty$. The algorithm uses $p,D,\eta$, but needs neither
$h$ nor a bound on the tree height; chains are implemented by counters.
The exact query schedule is specified in
Appendix~\ref{app:noisy-predecessor}.
\end{theorem}

The running-time bound is the barrier principle of
Section~\ref{sec:promise-sat} applied to the distance from the correct
endpoint, which starts at $2h+L_\eta$. The error bound uses one cover:
before any incorrect output, the walk must pass a live state whose
distance from the correct endpoint exceeds its starting distance by
more than $L_\eta$, and the harmonic first-hit bound prices that
excursion at $D\lambda^{-L_\eta-1}\le\eta$, with no union bound over
leaves. The correct
endpoint is used only in the proof. For a balanced tree
$h=O(\log(n+1))$; on an unbalanced tree the guarantee depends on the
queried leaf's depth alone.

\section{When restarting restores exponential tails}
\label{sec:prefix-restarts}

A restart changes the algorithm's state, but does not replenish the
randomness in the tape. A prefix cover nevertheless tells us when to
abandon an attempt: restart when the observed word reaches the cover.
The criterion below shows that robust almost-sure termination and an
exponential tail are possible under exactly the same strict cover
inequality, and that nothing can be done when it fails.

Fix an initial configuration $x$ and a deterministic algorithm that
consumes one source symbol per step. A restart resets the algorithm to
$x$ and resumes with the next unused symbol. The policy is
deterministic, uses only observed symbols, and may restart at most once
between successive steps. Running time counts source-consuming steps;
reset costs are separate. Write $\mathcal N_x$ for the original
nontermination event and $\alpha_s(x)=\mathcal C_s(\mathcal N_x)$,
where $0<s\le1$.

\begin{theorem}[Exact restart criterion]
\label{thm:prefix-restart-criterion}
The following statements are equivalent:
\begin{enumerate}[label=(\roman*),nosep]
\item $\alpha_s(x)<1$.
\item Some restart policy and constants $K<\infty$, $0<\gamma<1$ satisfy
$\Prb_\nu\{\tau_{\rm restart}>t\}\le DK\gamma^t$ for every $D\ge1$,
$\nu\in\mathcal Q_{s,D}$, and integer $t\ge0$.
\item Some restart policy terminates almost surely for every
$\nu\in\mathcal Q_{s,1}$.
\end{enumerate}
If $\alpha_s(x)=1$, every restart policy has nontermination probability one
under some source in $\mathcal Q_{s,1}$.
\end{theorem}

\begin{proof}[Proof of sufficiency]
If $\alpha_s(x)<1$, the continuity
\eqref{eq:prefix-cover-infinite-limit} of finite-horizon capacities
gives a horizon $T$ and a finite prefix-free cover $\mathcal B$ of
$\{\tau>T\}$, with
words of length at most $T$, whose cost $c=\sum_{b\in\mathcal B}P[b]^s$
is less than one. Remove words on which termination has already occurred.
Restart on reaching a word in $\mathcal B$. Every attempt ends within $T$
steps. The words describing $j$ failed attempts form the prefix-free family
$\mathcal B^j$, of total cost $c^j$. Consequently,
\[
 \Prb_\nu\{\tau_{\rm restart}>jT\}\le
 D\sum_{w\in\mathcal B^j}P[w]^s=Dc^j.
\]
This applies the prefix bound to the whole execution, across attempts,
with no conditional guarantee after a restart. The converse is proved
in Appendix~\ref{app:prefix-restarts}.
\end{proof}

For the converse, if $\alpha_s(x)=1$ the capacity identity gives a
source $\mu$ with $D=1$ supported on nontermination, and against any
restart policy the adversary restarts $\mu$ along with the algorithm:
the concatenation of independent copies of $\mu$, one per attempt,
still satisfies the prefix bound, and every attempt follows a
nonterminating execution until the policy abandons it.

\paragraph{Synthesis for explicit finite-state processes.}
For a finite explicit state space, one can synthesize the updates as
well as the restart rule; the problem is minimax control against a
source constrained along the whole trajectory rather than at each
state. Let $X$ be the live states, let $F_y$ be a
finite nonempty set of available updates at $y$, and let $\mathcal T_f(y,a)$
be the state reached when update $f$ reads symbol $a$. Set
\[
 R_f(y,z)=\sum_{a:\mathcal T_f(y,a)=z}\pi(a)^s,\qquad y,z\in X,
\]
omitting terminated states; $R_f(y,\cdot)$ is the row of weights that
update $f$ contributes to the cover recursion at $y$.
Updates are selected before reading the next symbol; selectors have no
additional private randomness. In the linear program below, the
constraints say that $v$ is a subsolution of the cover recursion under
every update, and maximizing picks the largest subsolution, which is
the infinite-horizon value: $V(y)$ is the worst-case nontermination
capacity from $y$ under the best selector.

\begin{theorem}[Synthesis of robust selection and restarts]
\label{thm:finite-state-restart-lp}
Let $V$ be the unique optimizer of
\begin{equation}\label{eq:restart-capacity-lp}
 \max\left\{\sum_{y\in X}v(y):
 0\le v\le\mathbf1,\quad
 v(y)\le\sum_{z\in X}R_f(y,z)v(z)
 \ \text{for every }y\in X,\ f\in F_y\right\}.
\end{equation}
\begin{enumerate}[label=(\roman*),leftmargin=*]
\item \textup{(Optimal selection.)}
There is a deterministic stationary selector $\mu$ such that, for every
initial state $x$ and every $D\ge1$,
\[
 \max_{\nu\in\mathcal Q_{s,D}}
     \Prb_\nu^{\mu,x}\{\tau=\infty\}
 =\min_{\Lambda}\max_{\nu\in\mathcal Q_{s,D}}
     \Prb_\nu^{\Lambda,x}\{\tau=\infty\}
 =\min\{1,DV(x)\},
\]
where $\Lambda$ ranges over deterministic nonanticipating selectors
without restarts.
\item \textup{(Restarting.)}
Fix a reset state $x_0$. Some deterministic nonanticipating policy,
allowed to select updates and restart to $x_0$, terminates almost surely
for every source in $\mathcal Q_{s,1}$ if and only if $V(x_0)<1$.
When $V(x_0)<1$, using $\mu$ and restarting on entry to
$S=\{y:V(y)=1\}$ gives constants $K<\infty$ and $0<\gamma<1$ such that
$\Prb_\nu\{\tau_{\rm restart}>t\}\le DK\gamma^t$ for every
$D\ge1$, $\nu\in\mathcal Q_{s,D}$, and integer $t\ge0$.
When $V(x_0)=1$, every such policy has nontermination probability one
under some source in $\mathcal Q_{s,1}$.
\item \textup{(Computation.)}
If the powered coefficients $R_f(y,z)$ are rational, then $V$, $\mu$,
the restart set $S$, and explicit constants for the tail bound are
computable in polynomial time in the explicit state--action description.
\end{enumerate}
\end{theorem}

The linear program \eqref{eq:restart-capacity-lp} is the classical
one of affine monotone dynamic programming
\citep{bertsekas-2019-affine}; here it yields exact guarantees against
globally constrained tapes, including an impossibility result for all
history-dependent policies.

The first linear program determines the optimal values, but an update
that merely satisfies the cover recursion need not attain them: at a
state with $V=0$, an update that always terminates and one that loops
forever both satisfy it. A second linear program resolves these ties on
the states with $V<1$. It produces a positive potential $\psi$ with
$\psi=\mathbf1+A\psi$, where $A$ is the selected powered transition
matrix on those states and $\mathbf1$ is the all-ones vector. Thus
$A\psi<\psi$, which forces geometric decay, and an adjusted potential
proves decay after transitions into $\{V=1\}$ are redirected to the
reset state (Appendix~\ref{app:prefix-restarts}). An implicit state
space can be exponentially large; the Hamming-ball rule of
Theorem~\ref{thm:promise-hamming-restarts} avoids the enumeration by
using the geometry of the assignments.

The rationality assumption in (iii) is mild: rounding each $\pi(a)^s$
up to a rational keeps the guarantees of (i) and (ii) for
$\mathcal Q_{s,D}$ with a computed value $\bar V\ge V$ that converges
to $V$ as the rounding is refined, though not the exact equalities
(Appendix~\ref{app:prefix-restarts}).

\paragraph{Adaptive versus predetermined restarts.}
Classical restart schedules choose cutoff times in advance
\citep{luby-sinclair-zuckerman-1993}; dynamic policies are studied in
\citep{kautz-et-al-2002-restarts,ruan-horvitz-kautz-2002-restarts}.
Under a shared tape, the time to abandon an attempt may need to depend
on what has been observed:
Proposition~\ref{prop:restart-oblivious-separation} gives a
three-state procedure for which every predetermined cutoff sequence
has worst-case nontermination probability at least $1/2$ in
$\mathcal Q_{1/2,1}$, whereas one stationary rule has exact worst-case
tail $\min\{1,D\,T2^{1-T}\}$: the source can cheaply hold the walk in
one state for any prescribed time, a predetermined cutoff tells it how
long, and restarting on entering that state never gives it the chance.

\section{Conclusion and further questions}\label{sec:related-work}
The independent-source stopping-time law does not determine worst-case
behavior under a prefix probability bound. Covers of varying depth
recover exactly what the law loses, and what they recover is
algorithmic: the order in which to attempt tasks, the source threshold
of a local walk, and the moment to restart.

Within the model, five questions remain: whether the noisy-search
threshold $s_c$ is sharp; whether a single SAT walk can match the
replay guarantee without storing the tape; whether compact state
representations can replace the enumeration in
Theorem~\ref{thm:finite-state-restart-lp}, as the Hamming ball does in
Theorem~\ref{thm:promise-hamming-restarts}; whether there are natural
algorithms for which covers using any fixed number of depths are
asymptotically weaker than general covers; and what entropy guarantee
a walk on an expander retains under an admissible source, which for
every $s<1$ can return the walk to a fixed vertex at a positive
fraction of times, whereas almost Chor--Goldreich sources let it
accumulate entropy \citep{doron-et-al-2023}.
Beyond it, every analysis assumes a known model and a fixed reference
law: what allows learning a policy that reaches a target reliably when
the model is unknown, and can covers still guide branching and
restarting when proposal distributions adapt to the observed history,
as in search with verified candidates?

\clearpage
\appendix

\section{Detailed related work}
\label{app:related-work}
Our source class belongs to the study of computation with weak
randomness. The Chor--Goldreich sources \citep{chor-goldreich-1988} bound each
block's min-entropy conditionally on all previous blocks; the almost
Chor--Goldreich sources of Doron, Moshkovitz, Oh, and Zuckerman
\citep{doron-et-al-2023} let a bounded number of block positions, fixed in
advance, be arbitrary and the remaining blocks be close to the required conditional
min-entropy, and are used there to drive random walks on expanders that
accumulate entropy and yield online condensers. Our condition is
cumulative rather than positional: which symbols are predictable may
depend on the history, as long as every prefix stays within its
budget.

The same authors
\citep[Definition~2]{doron-et-al-2025-online} study
\emph{unpredictable sources}: averaged over positions and draws of the
source, the realized next symbol has large conditional probability
only rarely. Their model permits predictable steps and includes
arbitrary finite sources of sufficiently high min-entropy rate, for
suitable block parameters. They analyze entropy accumulation on
lossless expanders and construct online condensers. Our source class
overlaps with this setting. For a uniform binary reference it requires
$H_\infty(X_{1:t})\ge st-\log_2D$ at every length $t$, so finite
truncations have high min-entropy rate when $s$ is close to one and $t$
is sufficiently large. A bound at one terminal horizon need not
constrain short prefixes this way. We determine exact survival
probabilities and use them to choose retry orders and restarts on the
same tape. Simulating a randomized algorithm with a weak source goes
back to Vazirani and Vazirani \citep{vazirani-vazirani-1985} for
semi-random sources and to Zuckerman \citep{zuckerman-1996-simulating}
and Andreev, Clementi, Rolim, and Trevisan \citep{andreev-et-al-1999}
for general min-entropy sources; these run the algorithm on many
strings derived from the source and combine the outcomes, whereas we
analyze one run that reads the source directly. Related work on
Markov-chain robustness controls changes to individual transition
probabilities \citep{zuckerman-2026-robustness}.

The source-capacity identity is classical tree max-flow/min-cut
\citep{bishop-peres-2017}; its dimension interpretation draws on
connections among symbolic spaces, entropy, and Hausdorff dimension
\citep{staiger-1993,mauldin-williams-1988,simpson-2015}.
Lutz \citep{lutz-2011} studies constructive and finite-state versions
of Billingsley dimension, which measure dimension relative to a
reference law. Our covers and source laws carry no computability
restriction. Achlioptas, Iliopoulos, and Vlassis
\citep[Section~4]{achlioptas-iliopoulos-vlassis-2017} cut each
trajectory when its probability crosses a threshold, using prefixes of
different lengths in their entropy-compression analysis. The capped
recurrence has the familiar multiplicative optimal-stopping form
\citep{jelito-pitera-stettner-2021}, and the finite-state synthesis
uses affine monotone dynamic programming \citep{bertsekas-2019-affine}.
In our setting the capacity
identity~\eqref{eq:exact-prefix-cover-robustness} gives exact source
probabilities and guides updates and restarts.

Two perspectives from learning theory motivated this formulation:
the hierarchy of local descriptions behind the pointwise generalization
bounds of Li and Xu \citep[Section~3.3]{li-xu-2026}, and the question,
formalized in \citep[Definition~2.2 and Proposition~2.3]{xu-2026-bellman},
of which history summaries preserve what a Bellman recursion needs.
The analogous question in this paper is which execution summaries
suffice for control; separate summaries at every depth can lose the way
prefixes extend one another. The connection is conceptual; our proofs
use the prefix-source constraint and tree capacity.

The promise-SAT walk and its distance analysis are due to Austrin,
Guruswami, and H{\aa}stad
\citep[Proposition~6.1]{austrin-guruswami-hastad-2017}.
Search within Hamming balls is also established
\citep{dantsin-hirsch-wolpert-2003}; here repeated resets return to the
same center under one cumulative tape constraint. The noisy-search
precedents are discussed in Section~\ref{sec:noisy-search}.

The cubic scale and its crossover have precedents in probability,
statistical physics, and algorithmic search. The asymptotic analysis
of Derrida and Simon \citep{derrida-simon-2007} predicts critical
survival of order $\exp[-\Theta(t^{1/3})]$ from fixed positive distance
for branching with an absorbing wall, with a near-critical crossover
governed by approximate boundary motion of the form
$m'(t)=\epsilon+c/m(t)^2$. Lemma~\ref{lem:promise-moving-plateau} proves the probability-one
plateau directly, in the discrete capped form with the near-critical
slack that the source construction needs; at criticality, Mallein's
consistent-displacement theorem
\citep[Theorem~1.1]{mallein-2019-consistent} also implies it, as
explained after that lemma. Berestycki, Liu, Mallein, and
Schweinsberg \citep{berestycki-et-al-2024-transition} prove a related
transition between cube-root and linear displacement at time of order
$\epsilon^{-3/2}$ for branching Brownian motion. Killed-branching
estimates also enter computational complexity: Pemantle
\citep{pemantle-2009-search} uses them to bound the cost of finding
nearly optimal paths in random trees. Our contribution is the sharp
source-robust characterization of the specified SAT walk. The upper
bound allows arbitrary clause overlap and history-dependent selection;
the matching construction gives one admissible source forcing the
delay with probability one for every selector. The proof establishes
exact saturation of the discrete capped recursion and uses the
capacity identity to construct the source.

\section{Prefix covers and worst-case sources}
\label{app:prefix-cover-proof}
Throughout this appendix, $P=\pi^{\mathbb N}$ is the full-support finite
source from Section~\ref{sec:prefix-cover-robustness}.  The source class
$\mathcal Q_{s,D}$ imposes its bound on every prefix of the full tape.
Symbols after termination may be continued with the reference law.

\paragraph{Remarks on the source class.}
Replacing the repeated block of the introductory stream by $m$ fixed
symbols preserves the prefix bound and makes the coordinates mutually
independent, so a lack of fresh randomness can arise with or without
dependence. Independent bits with bias $p$ satisfy the fair-bit version
of the condition whenever $\max\{p,1-p\}\le2^{-s}$. The bound is
information-theoretic, not a test for pseudorandomness: a generator
with at most $2^b$ seeds has, at every length, some output prefix of
probability at least $2^{-b}$, so a fixed finite seed cannot satisfy
the condition forever. This limits which sources the theorems cover; it
is not a lower bound on the randomness needed by every SAT solver, and
finite-time conclusions require the prefix inequalities only through
the deadline in question.

\begin{proof}[Proof of Theorem~\ref{thm:prefix-cover-robustness}]
Every prefix cover gives
$\nu(E)\le D\sum_{w\in\mathcal B}P[w]^s$, which proves the upper bound.
Suppose first that $E$ is determined by $T$ symbols.  In the full source
tree through depth $T$, give the root capacity $c(\varnothing)=1$ and every other node
$w$ capacity $c(w)=D P[w]^s$.  Route flow from the root only into the
depth-$T$ words whose cylinders lie in $E$.  The tree max-flow/min-cut
theorem gives a largest admissible flow $f$ of value
$\min\{1,D\mathcal C_s(E)\}$.
Indeed, cuts are exactly prefix covers, with root cost $1$; a cover of a
whole depth-$T$ cylinder by longer words cannot lower its cost, since
for nonnegative $x_i$ and $s\le1$,
$\sum_i x_i^s\ge(\sum_i x_i)^s$.

It remains to extend $f$ to a probability law on all depth-$T$ words.
Extend $f$ by zero on branches outside its support.  For every internal
node, including the root,
\[
 \sum_a c(wa)\ge c(w).
\]
For nonroot nodes this follows from $\sum_a \pi(a)^s\ge1$; at the root it
also uses $D\ge1$.  Since $f$ conserves flow, the residual capacities
$r(w)=c(w)-f(w)$ satisfy
\[
 \sum_a r(wa)\ge r(w).
\]
Start an additional flow $g$ of value $1-f(\varnothing)$ at the root.
Whenever its incoming mass at $w$ is at most $r(w)$, distribute that
mass among its children without exceeding their residual capacities.
The displayed inequality makes this possible at every step.
Then $f+g$ has total mass one and respects all node capacities.
It cannot give $E$ more than the maximal value of $f$, because its
restriction to successful leaves would itself be a larger admissible
flow.  Thus this probability law attains the upper bound.
Continue it beyond depth $T$ using the reference conditional law.  For
any suffix $v$,
\[
 \nu[wv]=\nu[w]P[v]
 \le D P[w]^sP[v]^s=D P[wv]^s,
\]
so it belongs to $\mathcal Q_{s,D}$ on the infinite tape space.

For a closed $E$, let $E_n$ be the union of the length-$n$ cylinders that
intersect $E$.  These are decreasing closed-and-open sets with
intersection $E$.  Their contents decrease to $\mathcal C_s(E)$:
monotonicity gives one inequality, while any finite prefix cover of
$E$ covers all sufficiently large $E_n$, by compactness.
The class $\mathcal Q_{s,D}$ is nonempty and weakly compact.  Indeed,
$P\in\mathcal Q_{s,D}$, the tape space is compact, and each cylinder
constraint is closed under weak convergence.
Choose maximizers $\nu_n$ for $E_n$ and a weakly convergent subsequence
with limit $\nu_*\in\mathcal Q_{s,D}$.  For fixed $m$ and $n\ge m$,
\[
 \nu_n(E_m)\ge\nu_n(E_n)
 =\min\{1,D\mathcal C_s(E_n)\}.
\]
Pass to the subsequential limit, using that $E_m$ is closed and open,
and then let $m\to\infty$.  Continuity from above yields
$\nu_*(E)\ge\min\{1,D\mathcal C_s(E)\}$, proving attainment and
\eqref{eq:exact-prefix-cover-robustness}.
For any decreasing sequence of finite-symbol events $E_T$ with
intersection $\Ncal$, the same compactness argument shows
$\mathcal C_s(E_T)\downarrow\mathcal C_s(\Ncal)$ and proves
\eqref{eq:prefix-cover-infinite-limit}.
\end{proof}

\begin{corollary}[Dimension as a source threshold]
\label{cor:dimension-source-threshold}
Let $\Ncal$ be closed and put $\delta=\dim_H^P\Ncal$.  For
$0<s<\delta$, some finite $D$ and some $\nu\in\mathcal Q_{s,D}$ satisfy
$\nu(\Ncal)=1$.  For $\delta<s\le1$, every such law gives
$\nu(\Ncal)=0$.  For every $0<s\le1$, the least domination constant permitting
a law supported on $\Ncal$ is
\begin{equation}\label{eq:minimal-domination-constant}
 \inf\{D\ge1:\exists\nu\in\mathcal Q_{s,D},\ \nu(\Ncal)=1\}
 =\frac1{\mathcal C_s(\Ncal)},
\end{equation}
where $1/0=\infty$.  If $P(\Ncal)=0$, then
\begin{equation}\label{eq:dimension-source-threshold}
 \delta=
 \inf\{s\in(0,1]:
       \nu(\Ncal)=0\text{ for every finite }D\ge1
       \text{ and every }\nu\in\mathcal Q_{s,D}\}.
\end{equation}
At $s=\delta>0$, there exists a law with finite domination constant
supported on $\Ncal$ exactly when $\mathcal C_\delta(\Ncal)>0$; dimension alone does
not settle this endpoint.
\end{corollary}

\begin{proof}[Proof of Corollary~\ref{cor:dimension-source-threshold}]
In the metric $d_P(\omega,\omega')=P[\omega\wedge\omega']$, each cylinder
$[w]$ has diameter $P[w]$.  Any nonsingleton set is contained in a
cylinder with the same diameter; singleton sets can be covered by
cylinders of arbitrarily small diameter.  Thus the infimum over
countable cylinder covers equals the Hausdorff content.  For compact
$\Ncal$, every such cover has a finite subcover, so this infimum is
$\mathcal C_s(\Ncal)$.  Its zero threshold is therefore
$\dim_H^P\Ncal$: the content is positive for $s<\delta$ and zero for
$s>\delta$.
Theorem~\ref{thm:prefix-cover-robustness} proves
\eqref{eq:minimal-domination-constant} and the strict-threshold claims.
If $P(\Ncal)=0$, then $\mathcal C_1(\Ncal)=P(\Ncal)=0$, so the set in
\eqref{eq:dimension-source-threshold} is nonempty and has infimum
$\delta$.  The same exact identity gives the stated critical criterion.
\end{proof}

\subsection{Controlled prefix capacities}
\label{app:controlled-prefix-statement}
For $0<s\le1$, take a finite state space, a common finite symbol
alphabet with full-support probabilities $\pi(a)$, and a finite nonempty
set $F_x$ of available updates at each live state. Write
$P=\pi^{\mathbb N}$ and let $\mathcal T_f(x,a)$ be the state obtained when
update $f$ uses symbol $a$ at state $x$. Set $V_0(x)=1$ at live states
and $V_t(x)=0$ at terminated states. At live states define
\begin{equation}\label{eq:controlled-prefix-recursion}
 V_{t+1}(x)=\min\left\{1,\min_{f\in F_x}
              \sum_a\pi(a)^sV_t(\mathcal T_f(x,a))\right\}.
\end{equation}
\begin{proposition}[Optimal deterministic selection]
\label{prop:controlled-prefix-value}
For every $D\ge1$ and finite horizon $T$,
\begin{equation}\label{eq:controlled-prefix-value}
 \min_{\Lambda}\max_{\nu\in\mathcal Q_{s,D}}
       \Prb_\nu^\Lambda\{\tau>T\}=\min\{1,DV_T(x_0)\},
\end{equation}
where $\Lambda$ ranges over deterministic nonanticipating selectors.
One selector depending only on current state and remaining time is
optimal simultaneously for every $D\ge1$ and every initial state.
\end{proposition}

This recursion evaluates a single globally constrained source; it does
not impose independent conditional constraints after successive symbols.

\begin{proof}[Proof of Proposition~\ref{prop:controlled-prefix-value}]
Fix a finite horizon.  At a live root, a prefix cover either cuts the
root at cost one or covers the children after the chosen update.
For a specified first update $f$, the selector can choose its continuation
independently after each observed symbol $a$.  The optimal continuation
cost is $V_t(\mathcal T_f(x,a))$, multiplied by the edge cost $\pi(a)^s$.
Minimizing over $f$ and comparing with the root cut gives
\eqref{eq:controlled-prefix-recursion} by induction.  At the horizon,
the surviving event costs one at a live state and zero at a terminated
state, exactly the stated boundary values.  Each minimum is attained,
and an attaining choice can depend only on current state and remaining
time.  These choices depend on $s$ but not on $D$, and the same
state-based rule is available from every initial state.
Theorem~\ref{thm:prefix-cover-robustness} now turns the optimal
cover cost into \eqref{eq:controlled-prefix-value}, since
$c\mapsto\min\{1,Dc\}$ is nondecreasing.

This proof optimizes one globally constrained source.  It does not reset
the domination constant after a prefix or give the source independent
conditional choices at each step.  The selector has no additional private
randomness; allowing trusted private coins would define a different
optimization problem.
\end{proof}

\section{Proofs for the promise-SAT walk}
\label{app:promise-proofs}

\subsection{A barrier bound from an integer potential}
The proof has two parts. A harmonic function pays for paths when they
first reach a distant barrier. A spectral estimate pays for paths that
remain below it. Both bounds are on prefix weights; neither treats the
actual source as a Markov chain.

The upper bound uses only a property of the potential after each
possible source symbol.  In particular, it needs neither a finite state
space nor a Markov transition law for the potential.

\begin{lemma}[Integer-potential barrier]
\label{lem:promise-potential-barrier}
Fix $k/2<r<k$ and the parameters $\lambda,s_*$ in
\eqref{eq:promise-parameters}.  Consider a deterministic process driven
by symbols in $[k]$, with uniform product reference law $P$ and stopping
time $\tau$.  Assign to every live history a positive integer potential
$d$, initially $h\ge1$, and give every terminated history potential zero.
At each live history, assume that at least $r$ of the $k$ next symbols
either terminate or give a live potential at most $d-1$, and that every
other symbol either terminates or gives a live potential at most $d+1$.
For every integer $m>h$, every $T\ge0$, every $D\ge1$, and every
$\nu\in\mathcal Q_{s_*,D}$,
\begin{equation}\label{eq:promise-explicit-barrier}
 \Prb_\nu\{\tau>T\}
 \le \min\left\{1,D\left[
 \frac hm\lambda^{h-m}
 +\frac{\lambda^h}{\sqrt{\lambda^2-1}}
   \exp\!\left(-\frac{\pi^2T}{2(2m-1)^2}\right)
 \right]\right\}.
\end{equation}
Consequently the critical tail and cubic expectation bounds in
Theorem~\ref{thm:promise-robust} hold.
\end{lemma}

\begin{proof}
Put
\[
 a=rk^{-s_*}=\lambda/2,
 \qquad b=(k-r)k^{-s_*}=1/(2\lambda).
\]
If $f$ is nonnegative and increasing, with $f(0)=0$, the sum of
$f$ over the live children of a history of potential $d$, multiplied by
the powered edge weight $k^{-s_*}$, is at most
\begin{equation}\label{eq:promise-monotone-comparison}
 a f(d-1)+b f(d+1).
\end{equation}
Indeed, assign $r$ of the decreasing symbols to the first term.  Each
other live child has potential at most $d+1$; terminated children
contribute zero.  This argument applies separately after every history.

Cover $\{\tau>T\}$ by prefixes that first reach potential $m$ while
still live, up to time $T$, and live length-$T$ prefixes that never
reach $m$.  Remove any length-$T$ prefix already in the first set.
Because upward increments are at most one, the first set meets every
path that leaves $\{1,\ldots,m-1\}$.  The resulting cover is prefix-free.

For first-hit prefixes use
\[
 H_m(d)=\frac dm\lambda^{d-m}\qquad(0\le d\le m).
\]
This function is increasing, $H_m(0)=0$, $H_m(m)=1$, and
\[
 aH_m(d-1)+bH_m(d+1)=H_m(d)\qquad(1\le d<m).
\]
Start with weight $H_m(h)$ at the root.  Expanding a live node below
$m$ cannot increase the sum of its powered prefix weight times $H_m$,
by \eqref{eq:promise-monotone-comparison}.  Leave first-hit nodes
unexpanded, give terminated nodes value zero, and continue to depth
$T$.  Every first-hit node has value one; the remaining values are
nonnegative.  Finite-tree induction therefore bounds the weight
of first-hit prefixes by $H_m(h)$.

For paths remaining below $m$, put $L=m-1$ and define a reflected
comparison matrix $R^{(m)}$ on $\{1,\ldots,L\}$ by
\[
 (R^{(m)}f)(d)=a f(d-1)+b f(d+1)\quad(d<L),\qquad
 (R^{(m)}f)(L)=a f(L-1)+b f(L),
\]
where $f(0)=0$.  This matrix preserves nonnegative increasing vectors:
consecutive differences in the interior are nonnegative linear
combinations of consecutive differences of $f$, and the final
difference is $a(f(L-1)-f(L-2))\ge0$ when $L\ge2$.
The claim is immediate for $L=1$.
Thus $F_t=(R^{(m)})^t\mathbf1$ is nonnegative and increasing.
Induction on the remaining horizon bounds the weight of
length-$T$ live paths that stay below $m$ by $F_T(h)$.
For a history with potential below $L$, apply
\eqref{eq:promise-monotone-comparison} to $F_t$.
At potential $L$, an upward exit is discarded, each of $r$ decreasing
symbols contributes at most $F_t(L-1)$, and any other retained child
contributes at most $F_t(L)$.  This gives the reflected final row.

Let $B=\operatorname{diag}(\lambda,\ldots,\lambda^L)$.
The symmetric matrix $B^{-1}R^{(m)}B$ has off-diagonal entries $1/2$,
last diagonal entry $b<1/2$, and other diagonal entries zero.
Increasing the last diagonal entry to $1/2$ gives a nonnegative
symmetric matrix with largest eigenvalue
$\cos(\pi/(2L+1))$.
To see the latter, the positive vector
$v_j=\sin(j\pi/(2L+1))$ satisfies its eigenvector equation,
including the reflected final row.
A symmetric nonnegative matrix has Euclidean operator norm equal to
its spectral radius.  Perron--Frobenius monotonicity consequently gives
\[
 \|B^{-1}R^{(m)}B\|_2\le\cos\frac{\pi}{2m-1}.
\]
Since $\|(\lambda^{-j})_{j=1}^L\|_2\le1/\sqrt{\lambda^2-1}$,
\[
 F_T(h)\le\frac{\lambda^h}{\sqrt{\lambda^2-1}}
 \left(\cos\frac{\pi}{2m-1}\right)^T
 \le\frac{\lambda^h}{\sqrt{\lambda^2-1}}
 \exp\!\left(-\frac{\pi^2T}{2(2m-1)^2}\right).
\]
The source constraint bounds the probability of each cover cylinder
by $D$ times its powered reference mass.  Summing the two parts proves
\eqref{eq:promise-explicit-barrier}.

When $T^{1/3}>h$, choose an integer $m>h$ within two of $T^{1/3}$.
Both terms are at most $C\lambda^h e^{-cT^{1/3}}$.
When $T^{1/3}\le h+1$, the same capped bound follows from probability
at most one, after enlarging $C$ and decreasing $c$.
Finally, summing $\min\{1,CD\lambda^h e^{-cT^{1/3}}\}$ over integer
$T\ge0$ gives $C'(h+\log D+1)^3$.
For example, split the sum at a sufficiently large constant times
$(h+\log D+1)^3$; beyond that point its summands are bounded by
$e^{-cT^{1/3}/2}$.
\end{proof}

\begin{proof}[Proof of Theorem~\ref{thm:promise-robust}]
Let $d$ be the Hamming distance to $x^*$. Whenever a clause is
violated, every literal it satisfies under $x^*$ belongs to a coordinate
on which the current assignment and $x^*$ disagree. At least $r$
flips therefore decrease $d$ by one; all others increase it by one.
Assign potential zero on termination. Every live assignment has
positive potential. The barrier lemma and
$\mathcal Q_{s,D}\subseteq\mathcal Q_{s_*,D}$ give the critical tail
and the cubic expectation bound. An initially satisfying assignment
has $\tau=0$.

For the additional bound above the threshold, write
$\delta=s-s_*$ and sum over live length-$t$ histories:
\[
 M_t=\sum_{\substack{|w|=t\\\tau(w)>t}}k^{-st}\lambda^{d(w)}.
\]
At each live history the sum of the children's potential weights is
at most
\[
 k^{-s}\bigl(r\lambda^{d-1}+(k-r)\lambda^{d+1}\bigr)
 =k^{-\delta}\lambda^d.
\]
Terminated children contribute zero. Thus $M_T\le\lambda^h k^{-\delta T}$
by finite-tree induction. Every live leaf has $\lambda^{d(w)}\ge1$,
so the prefix bound gives
$\Prb_\nu\{\tau>T\}\le D M_T$.
For $\delta>0$, summing the capped exponential tail gives
$O_{k,r}(L/\delta)$, where $L=h+\log D+1$. Combining this with
the cubic bound proves \eqref{eq:promise-crossover}.
This weighted-tree argument also applies to the general potential
in Lemma~\ref{lem:promise-potential-barrier}; it uses no conditional
probability assumption on the source.
Under $P$, the expected decrease in distance at a live history is at
least $(2r-k)/k$, giving $\E_P\tau\le kh/(2r-k)$ by additive drift.
\end{proof}

\paragraph{The single-depth baseline.}
For SAT, live distances are at most $n$.  Use the reflected matrix
above on $\{1,\ldots,n\}$.  At distance $n$, every actual flip
decreases distance, so its powered contribution on an increasing
vector $f$ is at most $(a+b)f(n-1)\le af(n-1)+bf(n)$.
The same induction and spectral estimate therefore bound the
single-depth live-prefix sum by
\begin{equation}\label{eq:promise-single-depth-baseline}
 \widetilde Z_T(s_*):=
 \sum_{\substack{w\in[k]^T:\tau(w)>T}}k^{-s_*T}
 \le \frac{\lambda^h}{\sqrt{\lambda^2-1}}
       \exp\!\left(-\frac{\pi^2T}{2(2n+1)^2}\right).
\end{equation}
Applying the source bound at depth $T$ and summing the capped tail
proves $\E_\nu\tau=O_{k,r}(n^2(h+\log D+1))$.

\subsection{A family with a sharp source threshold}
For $n\ge k+1$, let $\Phi_n$ contain every clause
\begin{equation}\label{eq:promise-complete-formula}
 \left(\bigvee_{i\in I}x_i\right)
 \lor\left(\bigvee_{j\in J}\neg x_j\right),
 \quad I\cap J=\varnothing,\quad |I|=r,\quad |J|=k-r.
\end{equation}
There are $\binom nk\binom kr=O_k(n^k)$ clauses.
The all-one assignment strongly $r$-satisfies the formula.
It is the unique such assignment. For an assignment with $z\ge1$
zeroes and $o=n-z$ ones, if $o\ge k-r$, put $k-r$ ones in the negative
positions and at least one zero in the positive positions; the clause
has fewer than $r$ true literals. If $o<k-r$, put all the ones in
negative positions and fill the remaining positions with zeroes;
the clause has $k-r-o<r$ true literals. The required choices exist
because $n\ge k$.

Order positive literals first in each clause.
If $d$ is the number of zero variables, a clause is violated precisely
when its positive variables are zero and its negative variables are
one.  Hence an assignment is live exactly when
$r\le d\le n-k+r$.
Put $N=n-k+2$ and $j=d-r+1$; the live values are $1,\ldots,N-1$.
For every selector, symbols $1,\ldots,r$ decrease $j$, and the other
symbols increase it.  The powered matrix of live transitions is
\begin{equation}\label{eq:promise-exact-matrix}
 R_s(j,j-1)=r k^{-s},\qquad
 R_s(j,j+1)=(k-r)k^{-s},
\end{equation}
with out-of-range entries omitted.

\begin{proof}[Proof of Theorem~\ref{thm:promise-sharp}(i)]
The spectral radius of \eqref{eq:promise-exact-matrix} is
\[
 \rho_s=2\sqrt{r(k-r)}\,k^{-s}\cos(\pi/N).
\]
For $s<s_*$, take $N$ sufficiently large that $\rho_s>1$.
A positive right eigenvector is
$\varphi(j)=\lambda^j\sin(j\pi/N)$.
Start at a maximizing $j_0$, and at live value $j$ give label $a$
probability
\[
 q(a\mid j)=k^{-s}\frac{\varphi(j+\delta(a))}{\rho_s\varphi(j)},
 \qquad
 \delta(a)=\begin{cases}-1,&a\le r,\\1,&a>r,\end{cases}
\]
where $\varphi(0)=\varphi(N)=0$.
The eigenvector equation normalizes these probabilities, and the
boundary labels have probability zero.
For every supported word $w$ of length $t$, telescoping gives
\[
 \nu[w]=k^{-st}\rho_s^{-t}
       \frac{\varphi(j_t)}{\varphi(j_0)}
 \le k^{-st}=P[w]^s.
\]
The resulting source belongs to $\mathcal Q_{s,1}$ and never exits the
live interval.  The path $j_t$ is determined by source labels alone,
independently of the clause selector.  Thus this one source forces
nontermination for all selectors simultaneously.
\end{proof}

\begin{corollary}[Full support and nonperiodic failure]
\label{cor:promise-full-support}
For every $0<s<s_*$ and $0<\varepsilon<1$, there are a strongly
$r$-satisfiable formula, a fixed initial assignment, and a single
source $\nu\in\mathcal Q_{s,1}$ such that
\[
 \nu\{\tau=\infty\}\ge1-\varepsilon
\]
for every deterministic nonanticipating clause selector. The source
can be generated by a time-homogeneous stochastic machine with
finitely many internal states, satisfies
$\nu[wa]/\nu[w]>0$ for every finite word $w$ and every symbol $a$,
and almost surely produces a tape that is not eventually periodic.
For each fixed deterministic memoryless clause selector, the
assignment trajectory on nontermination is almost surely not
eventually periodic.
\end{corollary}

The conditional probabilities need not have a common positive lower
bound. The generator is stochastic, not a deterministic finite-seed
PRNG. Nonperiodicity also does not mean that assignments are never
revisited: a finite assignment space forces revisits on every infinite
run. The statement rules out eventual repetition of a fixed cycle.
The proof uses the source from Theorem~\ref{thm:promise-sharp}(i).

\begin{proof}[Proof of Corollary~\ref{cor:promise-full-support}]
Let $q\in\mathcal Q_{s,1}$ be the source in the proof of
Theorem~\ref{thm:promise-sharp}(i), and set
$\nu=(1-\varepsilon)q+\varepsilon P$.
Both component laws satisfy the same prefix bound, so $\nu$ does
too. The $q$ component forces nontermination for all selectors
simultaneously, proving the probability bound.
The generator for $q$ has the finite internal state
$j\in\{1,\ldots,N-1\}$ and fixed transition probabilities;
an initial choice of mode implements the mixture with $P$ using
only finitely many additional states. Furthermore,
$\nu[wa]\ge\varepsilon k^{-|w|-1}>0$.

Every individual tape has mass at most
$\lim_{t\to\infty}k^{-st}=0$. Since the eventually periodic tapes
form a countable set, they have total mass zero. For a memoryless
selector, the current assignment determines the selected clause,
and the next assignment identifies the flipped coordinate and hence
its unique literal index. Thus an eventually periodic assignment
trajectory would determine an eventually periodic source tape.
\end{proof}

\subsection{Matching tails and the gap from single-depth covers}
\begin{proof}[Proof of Theorem~\ref{thm:promise-sharp}(ii)]
Write $R=R_{s_*}$, $a=\lambda/2$, $b=1/(2\lambda)$, and
$\epsilon=(s-s_*)\log k\ge0$, so that $R_s=e^{-\epsilon}R$.
The exact distance reduction gives the prefix-capacity recursion
\[
 C_0=\mathbf1,\qquad
 C_{t+1}=\min\{\mathbf1,e^{-\epsilon}RC_t\},
\]
with the minimum taken coordinatewise, and the prefix-cover identity
gives $A_{N,T}(s)=C_T(1)$.
The critical upper bound in Theorem~\ref{thm:promise-robust}, with
$h=r$, also applies to $\mathcal Q_{s,1}\subseteq\mathcal Q_{s_*,1}$:
$A_{N,T}(s)\le C_1e^{-c_1T^{1/3}}$.
Diagonal similarity transforms $R$ to the symmetric killed
simple-walk matrix on $1,\ldots,N-1$, whose Euclidean norm is
$\cos(\pi/N)$.  The full-depth cover therefore gives
\[
 A_{N,T}(s)\le e^{-\epsilon T}e_1^{\mathsf T}R^T\mathbf1
 \le\frac{\lambda}{\sqrt{\lambda^2-1}}
       e^{-\epsilon T-\pi^2T/(2N^2)}.
\]
The smaller of the two upper bounds is at most their geometric mean,
which gives $C_2\exp[-c_2(T^{1/3}+\epsilon T+T/N^2)]$.
Since $A_{N,T}(s)\le A_{N,1}(s)=e^{-\epsilon}b\le b<1$,
decreasing the exponent constant absorbs $C_2$ uniformly for
$N\ge3$, $T\ge1$, and $s\in[s_*,1]$.

For the lower bound, retain only paths within $1,\ldots,m-1$, where
$3\le m\le N$, and define
\[
 \rho_m=\cos(\pi/m),\qquad
 \varphi_m(j)=\lambda^j\sin(j\pi/m),\qquad
 u(j)=\frac{\varphi_m(j)}{\max_{1\le i<m}\varphi_m(i)}.
\]
Extend $u$ by zero outside this interval.  Then $0\le u\le1$ and
$Ru\ge\rho_m u$.  Since $e^{-\epsilon}\rho_m\le1$, induction
in the capped recursion gives
$C_T\ge e^{-\epsilon T}\rho_m^T u$; capping at one preserves
this lower bound at each step.  Furthermore,
\[
 u(1)\ge\frac2m\lambda^{2-m},\qquad
 \rho_m^T\ge e^{-\pi^2T/m^2}\quad(m\ge3).
\]
Consequently
\[
 A_{N,T}(s)\ge
 \frac2m\lambda^{2-m}e^{-\epsilon T-\pi^2T/m^2}.
\]
Choose $m=\min\{N,\max\{3,\lceil T^{1/3}\rceil\}\}$.
For $T\le N^3$ the last expression is at least
$e^{-C(T^{1/3}+\epsilon T)}$; for $T>N^3$ it is at least
$e^{-C(N+\epsilon T+T/N^2)}$.
Since $N<T/N^2$ in the latter case, this proves
\eqref{eq:promise-profile}, with $\epsilon=(s-s_*)\log k$.

For the single-depth comparison, set $s=s_*$ and write $S_N$ for
the symmetric killed simple-walk matrix.  Diagonal similarity gives
\[
 e_1^{\mathsf T}R^t\mathbf1
 =\sum_{j=1}^{N-1}\lambda^{1-j}S_N^t(1,j).
\]
The reflection principle on the positive half-line gives, when
parities match,
\[
 S_\infty^t(1,j)
 =2^{-t}\left[\binom{t}{(t+j-1)/2}
              -\binom{t}{(t+j+1)/2}\right]
 \le C\frac{j}{(t+1)^{3/2}}.
\]
For the inequality, factor the binomial difference as
$2j/(t+j+1)$ times its first term and use the central-binomial bound.
Since $S_N^t\le S_\infty^t$ and
$\sum_{j\ge1}j\lambda^{1-j}<\infty$, the single-depth sum is at most
$C(t+1)^{-3/2}$.

For even $t=2i\le N-2$, the Catalan paths from $1$ back to $1$
stay below $N$.  Their weight is
$\operatorname{Cat}_i(ab)^i
 =\operatorname{Cat}_i/4^i\ge c(t+1)^{-3/2}$.
For odd $t=2i+1\le N-2$, append one upward step to these paths;
this multiplies the mass by $b$ and gives the same order.
Thus the single-depth sum is $\Theta_{k,r}((t+1)^{-3/2})$
for every $0\le t\le N-2$.

It remains to compare covers of the same event $E_T=\{\tau>T\}$
at all depths, including depths exceeding $T$.
Every live state has a live neighbor when $N\ge3$, and hence has an
arbitrarily long live continuation.  For $t\le T$, the length-$t$
cylinders meeting $E_T$ are exactly the live words.  Therefore
$F_t(E_T)=e_1^{\mathsf T}R^t\mathbf1$.
For $t>T$, all $k^{t-T}$ extensions of each live length-$T$ word
belong to $E_T$, so
\[
 F_t(E_T)=k^{(1-s_*)(t-T)}e_1^{\mathsf T}R^T\mathbf1
 \ge F_T(E_T).
\]
Taking the infimum over all depths proves \eqref{eq:promise-flat-gap}.
\end{proof}

\subsection{Sharp dependence on distance and source slack}
The capped recursion retains a region where its value is exactly one.
At the critical source exponent its boundary advances on the scale
$t^{1/3}$; above the threshold there is an additional linear advance.
The exact unit plateau will give a single admissible source that forces
a prescribed delay with probability one.

In more detail, we construct a comparison profile whose boundary, at
distance $m$, advances by at most order $\epsilon+m^{-2}$ per step,
where $\epsilon=(s-s_*)\log k$. At
criticality the speed is of order $m^{-2}$, so reaching distance $j$
takes cubic time; a positive $\epsilon$ eventually gives speed of order
$\epsilon$ and time $j/\epsilon$. An exponentially tilted sine profile
makes this comparison rigorous. Until the boundary reaches the initial
distance, the capacity remains exactly one, already with $D=1$, and tree
max-flow supplies a source supported entirely on executions that survive
the required time. For the slack dependence, a separate cover lower
bound keeps capacity at least $1/D$ for the required horizon; the two
constructions give $L=h+\log D+1$ in
Theorem~\ref{thm:promise-distance-sharp}. A fixed profile at the far
end of the live interval then extends the plateau to the two-sided
recursion of the complete formulas
(Corollary~\ref{cor:promise-linear-distance}).

\begin{lemma}[A slowly advancing plateau]
\label{lem:promise-moving-plateau}
Let $\lambda>1$, $a=\lambda/2$, and $b=1/(2\lambda)$.
For $\epsilon\ge0$ define, on the positive integers,
\[
 C_0^\epsilon(j)=1,\qquad C_t^\epsilon(0)=0,\qquad
 C_{t+1}^\epsilon(j)=
 \min\{1,e^{-\epsilon}(aC_t^\epsilon(j-1)+bC_t^\epsilon(j+1))\}.
\]
There are constants $\epsilon_0,c>0$ and $j_0\ge1$, depending only
on $\lambda$, such that, for $0\le\epsilon\le\epsilon_0$,
\[
 C_t^\epsilon(j)=1\qquad\text{if}\qquad
 j\ge j_0,\qquad 0\le t\le c\frac{j^3}{1+\epsilon j^2}.
\]
In particular, the permitted horizon has order
$\min\{j^3,j/\epsilon\}$, with $j/0=\infty$.
\end{lemma}

The critical plateau order has an antecedent in
Mallein~\citep[Theorem~1.1]{mallein-2019-consistent}.
Indeed, put $\beta=\log\lambda$ and take a branching random walk
with bounded offspring, at least one child per particle, and expected
child counts $a,b$ at displacements $+\beta,-\beta$.
Its boundary identities are
$ae^{-\beta}+be^\beta=1$ and
$a\beta e^{-\beta}-b\beta e^\beta=0$.
The probability of a lineage staying below $j\beta$ through generation
$t$ is at most $C_t^0(j)$ by the union bound. Mallein's theorem makes
this probability tend to one for $j=\lceil At^{1/3}\rceil\pm1$ and
sufficiently large $A$; since $a+b>1$, the next capped step equals one.
The crossover at time of order $\epsilon^{-3/2}$ is also analogous to
the near-critical branching Brownian transition studied by
Berestycki, Liu, Mallein, and
Schweinsberg~\citep{berestycki-et-al-2024-transition}.
The proof below establishes the needed discrete capped-recursion
bound directly; its application here is the construction of sources
satisfying the prefix bound at every word.

\begin{proof}
Write $\beta=\log\lambda$. For each sufficiently large real $m$,
let $\theta_m\in(0,\pi/m)$ be the unique solution of
\[
 m\theta_m+\arctan(\theta_m/\beta)=\pi.
\]
Define a function on $[0,\infty)$ by
\[
 f_m(x)=
 \begin{cases}
 e^{\beta(x-m)}\dfrac{\sin(\theta_m x)}{\sin(\theta_m m)},
       &0\le x<m,\\[4pt]
 1,&x\ge m.
 \end{cases}
\]
This function increases from zero to one: its logarithmic derivative
below $m$ is $\beta+\theta_m\cot(\theta_m x)>0$, and that derivative
vanishes at $m$. Write $g_m$ for the sine expression without truncation
and $\rho_m=\cos\theta_m$. Direct calculation gives
\[
 a g_m(x-1)+b g_m(x+1)=\rho_m g_m(x).
\]
For $0\le z\le1$,
\[
 g_m(m+z)=e^{\beta z}
 \left(\cos(\theta_m z)-\frac{\beta}{\theta_m}
                    \sin(\theta_m z)\right)\le1;
\]
its derivative in $z$ is
$-e^{\beta z}(\beta^2/\theta_m+\theta_m)\sin(\theta_m z)\le0$.
Thus replacing $g_m$ by $f_m$ cannot decrease a neighbor value used
at an integer $j<m$. If
$F_\epsilon f(j)=\min\{1,e^{-\epsilon}(af(j-1)+bf(j+1))\}$, then
\begin{equation}\label{eq:plateau-eigen-comparison}
 F_\epsilon f_m(j)\ge e^{-\epsilon}\rho_m f_m(j)
 \qquad(1\le j<m).
\end{equation}

Choose fixed $z_0\in(0,1)$ and $\epsilon_0>0$ sufficiently small
that
\[
 e^{-\epsilon_0}
 \left[a e^{-\beta(z_0+1)}(1+\beta(z_0+1))+b\right]>1.
\]
Such a choice is possible because the bracket at $z_0=0$ equals
$(1+\beta+e^{-\beta})/2>1$.
For every fixed $z\ge0$,
$f_m(m-z)\longrightarrow e^{-\beta z}(1+\beta z)$.
Consequently, for all sufficiently large $m$ and uniformly in
$0\le\epsilon\le\epsilon_0$,
\begin{equation}\label{eq:plateau-saturated-band}
 F_\epsilon f_m(j)=1\qquad(j\ge m-z_0).
\end{equation}
Indeed, monotonicity gives
$f_m(j-1)\ge f_m(m-z_0-1)$, while $f_m(j+1)=1$.

Away from this saturated band, increasing $m$ decreases $f_m$ at a
uniform logarithmic rate. Differentiating the equation for $\theta_m$
gives
\[
 -\frac{\theta_m'}{\theta_m}
 =\frac1{m+\beta/(\beta^2+\theta_m^2)}.
\]
For fixed $0<x<m$ it follows that
\begin{equation}\label{eq:plateau-parameter-derivative}
 -\partial_m\log f_m(x)
 =\frac{\theta_m x\cot(\theta_m x)+\beta m}
             {m+\beta/(\beta^2+\theta_m^2)}.
\end{equation}
Since $y\cot y$ decreases on $(0,\pi)$, the right side for
$x\le m-z_0$ is bounded below by its value at $x=m-z_0$.
That value converges to $\beta^2z_0/(1+\beta z_0)>0$.
Hence there are $m_0\ge1$ and $\eta>0$ such that
$-\partial_m\log f_m(x)\ge\eta$ whenever
$m\ge m_0$ and $0<x\le m-z_0$.
Enlarge $m_0$ so that all the preceding comparisons hold and
$\rho_m>0$ for $m\ge m_0$.

For fixed $\epsilon\in[0,\epsilon_0]$, set
\[
 \Delta_m=\frac{\epsilon-\log\rho_m}{\eta}.
\]
For $j<m-z_0$, integrating
\eqref{eq:plateau-parameter-derivative} over $[m,m+\Delta_m]$ gives
\[
 f_{m+\Delta_m}(j)\le e^{-\epsilon}\rho_m f_m(j)
 \le F_\epsilon f_m(j).
\]
For the remaining $j$, use \eqref{eq:plateau-saturated-band}.
Thus $F_\epsilon f_m\ge f_{m+\Delta_m}$ pointwise.
Since $\theta_m<\pi/m$, there is a constant $A$ such that, uniformly
in $m\ge m_0$ and $\epsilon\in[0,\epsilon_0]$,
\[
 0<\Delta_m\le A(\epsilon+m^{-2}).
\]
Starting from $m_0$, define $m_{t+1}=m_t+\Delta_{m_t}$.
Monotonicity of $F_\epsilon$ and $C_0^\epsilon\ge f_{m_0}$ give
$C_t^\epsilon\ge f_{m_t}$ by induction.

We now bound this discrete sequence. Put
$Q=A(\epsilon_0/m_0+m_0^{-3})$ and $K=A(3+3Q+Q^2)$.
Because $\Delta_m/m\le Q$,
\[
 (m+\Delta_m)^3-m^3
 \le (3+3Q+Q^2)m^2\Delta_m
 \le K(1+\epsilon m^2).
\]
For any $j\ge2m_0$, suppose there is a first index $u\le t$ with
$m_u>j$. The preceding values are at most $j$, so summing the last
inequality yields
\[
 j^3<m_u^3\le m_0^3+Ku(1+\epsilon j^2)
 \le m_0^3+Kt(1+\epsilon j^2).
\]
This is impossible if
$t\le c j^3/(1+\epsilon j^2)$, where $c\le1/(2K)$, because
$m_0^3\le j^3/8$. Therefore $m_t\le j$ throughout that horizon and
$C_t^\epsilon(j)\ge f_{m_t}(j)=1$.
Taking $j_0=\lceil2m_0\rceil$ proves the lemma.
\end{proof}

\begin{proof}[Proof of Theorem~\ref{thm:promise-distance-sharp}]
Write $\delta=s-s_*$, $\epsilon=\delta\log k$, $\ell_D=\log D$, and
$j=h-r+1\ge1$. The source-label sequence determines the distance
process of $\Phi_n$ independently of the clause selector.
For an integer horizon $T\ge0$, choose $N=j+T+2$ and $n=N+k-2$.
No path starting at $j$ reaches the upper boundary by time $T$.
Since $R_s=e^{-\epsilon}R_{s_*}$, the finite-interval capacity
at $(j,T)$ therefore equals $C_T^\epsilon(j)$.
The finite-tree capacity identity, with root capacity one and all
nonroot capacities multiplied by $D$, gives
\begin{equation}\label{eq:promise-crossover-capacity}
 \max_{\nu\in\mathcal Q_{s,D}}\Prb_\nu\{\tau>T\}
 =\min\{1,D C_T^\epsilon(j)\}.
\end{equation}
The maximum is attained. A maximizing finite-depth law extends to
an infinite admissible source by appending independent uniform symbols:
for $t\ge T$, its prefix bound is
$Dk^{-sT}k^{-(t-T)}\le Dk^{-st}$ because $s\le1$.

First suppose $0\le\epsilon\le\epsilon_0$.
Lemma~\ref{lem:promise-moving-plateau} gives
$C_T^\epsilon(j)=1$ whenever $j\ge j_0$ and
\[
 T\le c\frac{j^3}{1+\epsilon j^2}.
\]
This supplies a probability-one delay of order
$\min\{j^3,j/\epsilon\}$ for large $j$.
To treat slack, monotonicity of the recurrence shows that
$C_t^\epsilon(j)$ is increasing in $j$.
The elementary inequality
$\min\{1,\alpha x\}\ge\alpha\min\{1,x\}$ for
$0\le\alpha\le1$ also gives, by induction,
\[
 C_T^\epsilon(1)\ge e^{-\epsilon T}C_T^0(1).
\]
The fixed-strip eigenvector lower bound used in
Theorem~\ref{thm:promise-sharp}(ii), with width
$\max\{3,\lceil T^{1/3}\rceil\}$, yields
\[
 C_T^\epsilon(1)
 \ge\exp[-\epsilon T-B(T+1)^{1/3}]
\]
for a constant $B$ depending only on $k,r$.
For all sufficiently large $\ell_D$, choose
\[
 T=\left\lfloor c'\min\{\ell_D^3,\ell_D/\epsilon\}\right\rfloor,
\]
again interpreting division by zero as infinity.
Choose $c'>0$ so that $c'\le1/4$ and $B(c')^{1/3}\le1/4$.
Then
\[
 \epsilon T+B(T+1)^{1/3}
 \le c'\ell_D+B(c')^{1/3}\ell_D+B\le\ell_D
\]
for sufficiently large $\ell_D$.
Thus $D C_T^\epsilon(j)\ge1$, and
\eqref{eq:promise-crossover-capacity} supplies another
probability-one delay, now of order
$\min\{\ell_D^3,\ell_D/\epsilon\}$.

These two horizons combine uniformly.
For $f_\epsilon(x)=\min\{x^3,x/\epsilon\}$,
monotonicity and $f_\epsilon(ux)\le u^3f_\epsilon(x)$ for $u\ge1$
show that the larger horizon, plus one to cover rounding and bounded
parameters, is at least $c_1f_\epsilon(h+\ell_D+1)$.
Here $h/r\le j\le h$, and $h\ge r$, so the change from $j$ to $h$
costs only a constant depending on $r$.
Because $\epsilon=(s-s_*)\log k$,
$f_\epsilon(L)=\Theta_k(\Psi_s(L))$ for $L=h+\ell_D+1$.
Reducing a fixed constant $c_2>0$ therefore gives an admissible source
with
\[
 \Prb_\nu\{\tau>\lfloor c_2\Psi_s(L)\rfloor\}=1.
\]
For bounded parameters the displayed horizon is zero and the
initial assignment is live; this also accounts for all floor terms.

It remains to treat $\epsilon\ge\epsilon_0$.
The target satisfies $\Psi_s(L)\le(\log k/\epsilon_0)L$.
Distance alone gives survival through $j-1$ steps under any source,
provided the upper boundary is farther than that horizon.
For slack, let
\[
 T_D=\left\lfloor\frac{\ell_D}{s\log k}\right\rfloor.
\]
Use a fixed upward label for the first $T_D$ symbols and independent
uniform symbols thereafter. Up to depth $T_D$, every nonzero cylinder
has probability one and satisfies $1\le Dk^{-st}$; after that depth,
the same extension calculation as above proves all prefix bounds.
With $N>j+T_D$ this source forces $\tau>T_D$.
The larger of $j-1$ and $T_D$, plus one, is at least a constant times
$h+\ell_D+1$. Reducing $c_2$ again proves the same probability-one
delay throughout the remaining range of $s$.

Finally take $T=\lfloor c_2\Psi_s(L)\rfloor$ and
$N=j+T+2$; any source constructed for a longer horizon can be
restricted to depth $T$ and extended uniformly.
The number of variables is
$n=j+T+k=O_{k,r}(\Psi_s(L))$, since $\Psi_s(L)\ge L\ge h$ for
$s\in[s_*,1]$.
Moreover $\tau>T$ implies $\tau\ge T+1>c_2\Psi_s(L)$ almost surely,
and hence gives the same lower bound on its expectation.
All claims hold simultaneously for every selector under the same
source because the label-to-distance map is selector independent.
\end{proof}

\begin{proof}[Proof of Corollary~\ref{cor:promise-linear-distance}]
Write $\epsilon=(s-s_*)\log k$ and $j=h-r+1$, and let $C^N_t$ be the
capped recursion of Theorem~\ref{thm:promise-sharp}(ii) on the live
values $1,\ldots,N-1$ of $\Phi_n$, that is, the recursion of
Lemma~\ref{lem:promise-moving-plateau} with the second boundary
condition $C^N_t(N)=0$. As in the proofs of
Theorem~\ref{thm:promise-sharp}(ii) and
Theorem~\ref{thm:promise-distance-sharp}, the worst-case probability of
$\{\tau>T\}$ over $\mathcal Q_{s,1}$ from the live value $j$ is
$\min\{1,C^N_T(j)\}$, it is attained by a source that extends to an
infinite admissible one, and the distance process is the same for
every selector. It therefore suffices to show that $C^N_T(j)=1$ for
the stated horizon.

\emph{A fixed profile at the far end.}
Let $\beta=\log\lambda$, so that $a+b=\cosh\beta$, and define
$\gamma\ge0$ by $\cosh\gamma=e^{\epsilon}$. The choice of $\epsilon_0$
in the proof of Lemma~\ref{lem:promise-moving-plateau} gives
$e^{-\epsilon_0}(a+b)>1$, because $e^{-y}(1+y)\le1$; hence
$\gamma<\beta$ for $\epsilon\le\epsilon_0$, and since
$e^{-\epsilon}(a+b)=k^{1-s}$, also $s_*+\epsilon_0/\log k<1$. Put
\[
 q(\ell)=\lambda^{-\ell}\,\frac{\sinh(\gamma\ell)}{\sinh\gamma}
 \qquad(\ell\ge0),
\]
read as $q(\ell)=\ell\lambda^{-\ell}$ when $\gamma=0$. Then $q(0)=0$,
$q>0$ on the positive integers, and the identity
$\sinh(\gamma(\ell+1))+\sinh(\gamma(\ell-1))=2\cosh\gamma\,\sinh(\gamma\ell)$
gives
\begin{equation}\label{eq:far-end-harmonic}
 e^{-\epsilon}\bigl(a\,q(\ell+1)+b\,q(\ell-1)\bigr)=q(\ell)
 \qquad(\ell\ge1).
\end{equation}
The ratio $q(\ell+1)/q(\ell)$ equals
$\lambda^{-1}(\cosh\gamma+\sinh\gamma\coth(\gamma\ell))$ for $\gamma>0$
and $\lambda^{-1}(1+1/\ell)$ for $\gamma=0$; since $\coth y\le1+1/y$
and $\sinh\gamma\le\gamma\cosh\gamma$, in both cases it is at most
$\lambda^{-1}(e^{\gamma_0}+\cosh\gamma_0/\ell)$, where
$\cosh\gamma_0=e^{\epsilon_0}$. Hence $q$ is nonincreasing beyond
$J^*=\lceil\cosh\gamma_0/(\lambda-e^{\gamma_0})\rceil$ and attains its
maximum over the positive integers at some $J\le J^*$, uniformly in
$0\le\epsilon\le\epsilon_0$. Define $g(\ell)=q(\ell)/q(J)$ for
$0\le\ell<J$ and $g(\ell)=1$ for $\ell\ge J$. Then $0\le g\le1$,
$g(0)=0$, and
\begin{equation}\label{eq:far-end-subsolution}
 g(\ell)\le\min\{1,e^{-\epsilon}(a\,g(\ell+1)+b\,g(\ell-1))\}
 \qquad(\ell\ge1).
\end{equation}
Indeed, for $\ell<J$ the right side equals $g(\ell)$ by
\eqref{eq:far-end-harmonic}; for $\ell=J$ the expression inside the
minimum is at least $e^{-\epsilon}(a\,q(J+1)+b\,q(J-1))/q(J)=1$ because
$q(J+1)\le q(J)$, so the right side is one; and for $\ell>J$ it equals
$\min\{1,e^{-\epsilon}(a+b)\}=1$.

\emph{Two boundary layers.}
The proof of Lemma~\ref{lem:promise-moving-plateau} constructs profiles
$f_m$ with $f_m(0)=0$, $0\le f_m\le1$, and $f_m(x)=1$ for $x\ge m$,
and a sequence $m_0\le m_1\le\cdots$ with $F_\epsilon f_{m_t}\ge
f_{m_{t+1}}$ pointwise on the positive integers, such that $m_t\le j$
whenever $j\ge j_0$ and $0\le t\le cj^3/(1+\epsilon j^2)$. Suppose
$j\ge j_0$ and $N-j\ge J^*+1$, and for such $t$ put
\[
 u_t(x)=\min\{f_{m_t}(x),\,g(N-x)\}\qquad(0\le x\le N),
\]
so that $u_t(0)=u_t(N)=0$ and $u_t(j)=1$. We claim that
$C^N_t\ge u_t$ on $\{1,\ldots,N-1\}$ throughout this horizon. This
holds at $t=0$. For the step from $t$ to $t+1$, fix $x$. If
$x+1\le N-J$, then $g(N-y)=1$ for $y\in\{x-1,x,x+1\}$, so $u_t=f_{m_t}$
at these points and monotonicity of the recursion gives
$C^N_{t+1}(x)\ge F_\epsilon f_{m_t}(x)\ge f_{m_{t+1}}(x)\ge u_{t+1}(x)$.
Otherwise $x\ge N-J$, so $x-1\ge N-J-1\ge j\ge m_t$, $f_{m_t}=1$ at
$x-1,x,x+1$, and $u_t(y)=g(N-y)$ there; then
\eqref{eq:far-end-subsolution} with $\ell=N-x$ gives
$C^N_{t+1}(x)\ge\min\{1,e^{-\epsilon}(a\,g(N-x+1)+b\,g(N-x-1))\}
\ge g(N-x)\ge u_{t+1}(x)$. Hence $C^N_t(j)\ge u_t(j)=1$, and the
worst-case survival probability is one, for every
$t\le cj^3/(1+\epsilon j^2)$.

\emph{Parameters.}
The condition $j\ge j_0$ reads $h\ge j_0+r-1$, and
$N-j=n-h-k+r+1\ge J^*+1$ reads $h\le n-k+r-J^*$; both hold for
$h_0\le h\le n-h_0$ with $h_0=\max\{j_0+r-1,J^*+k-r\}$, and such an
assignment is live. Since $h/r\le j\le h$ and
$1+\epsilon j^2\le2\max\{1,\epsilon j^2\}$, the horizon is at least
$\frac c2\min\{j^3,j/\epsilon\}\ge c'\min\{h^3,h/(s-s_*)\}=:H$ for a
constant $c'>0$ depending only on $k,r$; since $\tau>\lfloor H\rfloor$
forces $\tau\ge\lfloor H\rfloor+1>H$, this proves the almost-sure bound
with $\delta_0=\epsilon_0/\log k$.

For the last assertion take $s=s_*$, $h=\lfloor n/2\rfloor$, and
$n\ge2h_0$, so that $h_0\le h\le n-h_0$. The source constructed lies in
$\mathcal Q_{s_*,D}$ for every $D\ge1$ and gives $\E_\nu\tau\ge c'h^3$,
while Theorem~\ref{thm:promise-robust} gives
$\E_\nu\tau\le C(h+\log D+1)^3$ for every $\nu\in\mathcal Q_{s_*,D}$;
for fixed $D$ both are of order $n^3$. Under $P$ the expected number of
flips is at most $kh/(2r-k)$, and at least $\min\{j,N-j\}$ flips are
needed to leave the live interval, so it is of order $n$.
\end{proof}

\subsection{An observable restart rule for the promise-SAT walk}

\begin{proof}[Proof of Theorem~\ref{thm:promise-hamming-restarts}]
Use the unknown witness distance $d(x)=\operatorname{dist}(x,x^*)$
only in the analysis. Suppose a proposed flip reaches an unsatisfying
assignment $y$ outside the ball of radius $2H$ around $x_0$.
The triangle inequality gives
\[
 d(y)\ge \operatorname{dist}(y,x_0)-\operatorname{dist}(x_0,x^*)
 \ge 2H+1-h\ge h+1.
\]
Resetting to $x_0$ therefore strictly decreases this potential.
Consequently, after the flip and any reset, each of the at least $r$
favorable symbols still either terminates the walk or decreases $d$
by at least one; every other symbol increases $d$ by at most one.
Every live assignment kept by the algorithm lies within distance $2H$
of $x_0$, so its potential is at most $L=h+2H$.

Let $a=\lambda/2$ and $b=1/(2\lambda)$ as in
Lemma~\ref{lem:promise-potential-barrier}. On
$\{1,\ldots,L\}$ define the reflected comparison matrix
\[
 (Rf)(d)=af(d-1)+bf(d+1)\quad(d<L),\qquad
 (Rf)(L)=af(L-1)+bf(L),
\]
where $f(0)=0$. Its powers preserve nonnegative increasing vectors.
Exactly as in the barrier lemma, induction on the remaining horizon
therefore bounds the weight of surviving length-$T$ execution
words by $e_h^{\mathsf T}R^T\mathbf1$.
This comparison is valid after every history, including histories
containing resets; the clause selector may retain the entire history.
The spectral estimate in that lemma gives
\[
 e_h^{\mathsf T}R^T\mathbf1
 \le\frac{\lambda^h}{\sqrt{\lambda^2-1}}
 \exp\!\left(-\frac{\pi^2T}{2(2L+1)^2}\right).
\]
Apply the source constraint to whole length-$T$ execution words and
sum. Since $L=h+2H\le3H$, the resulting bound is at most
$CD\lambda^h e^{-cT/H^2}$, proving
\eqref{eq:promise-hamming-restart-tail} for $s=s_*$;
the inclusion $\mathcal Q_{s,D}\subseteq\mathcal Q_{s_*,D}$ gives
the claim for $s\ge s_*$.

Since $L\le3H$, summing the capped exponential tail yields
$O_{k,r}(H^2(h+\log D+1))$. More explicitly, for $A>0$ and $q>0$,
$\sum_{T\ge0}\min\{1,Ae^{-qT}\}$ is at most
$1+q^{-1}(1+\max\{0,\log A\})$ by comparison with the integral
of the decreasing summand. Here $q^{-1}=O(H^2)$ and
$\log A=h\log\lambda+\log D-\tfrac12\log(\lambda^2-1)$.

For the implementation claim, keep a list of coordinates on which the
current assignment differs from $x_0$, together with an inverse index
for constant-time insertion and deletion. A flip updates this list
in constant time. At a reset, restore and remove every listed
coordinate. Each restored coordinate was flipped at least once in
that attempt, so the reset uses no more variable updates than the
attempt used source-driven flips. Summing over disjoint attempts proves
the factor-two bound. The list has size at most $2H+1$ before a reset.
\end{proof}

\subsection{A radius-free portfolio using an archived tape}
\label{app:radius-free-replay}

\begin{theorem}[Unknown radius with an archived tape]
\label{thm:promise-radius-free-replay}
There is a deterministic algorithm that requires neither a distance
bound nor the source slack $D$ and has the following guarantee.  Suppose a strongly
$r$-satisfiable $k$-CNF on $n$ variables has a strongly satisfying
assignment at distance $h\ge1$ from the initial assignment.  Put
$\ell_h=1+\lceil\log_2h\rceil$ and $L=h+\log D+1$.
For every $D\ge1$, $s\in[s_*,1]$, and
$\nu\in\mathcal Q_{s,D}$, the algorithm's total number $\sigma$ of
variable updates satisfies, for every integer $T\ge0$,
\[
 \Prb_\nu\{\sigma>T\}
 \le \min\left\{1,CD\lambda^h
       \exp\!\left(-\frac{cT}{h^2\ell_h^2}\right)\right\},
\]
and
\[
 \E_\nu\sigma
 \le C\min\left\{\Psi_s(L),h^2\ell_h^2 L\right\}.
\]
Here $c,C>0$ depend only on $k,r$.  The update count includes every
flip in every simulation and every coordinate restored during a reset.
The algorithm keeps at most $2+\lceil\log_2n\rceil$ simulated
assignments and at most $\sigma$ archived source symbols.  Each
simulation may use any deterministic clause selector based only on
its own local history.  Work to find a violated clause is separate.
\end{theorem}

\paragraph{The algorithm.}
Stop immediately if the initial assignment $x_0$ satisfies the formula.
Let $J=\lceil\log_2n\rceil$.  Keep one base simulation of the original
walk and auxiliary simulations $A_j$, $0\le j\le J$, of the
Hamming-ball restart rule with radius parameter $H_j=2^j$.
All simulations start from $x_0$ and use the same source tape
$\omega_1,\omega_2,\ldots$; an auxiliary simulation can be initialized
when first scheduled.
Partition the nonnegative indices into groups
\[
 G_b=\{2^b-1,\ldots,2^{b+1}-2\},\qquad b\ge0.
\]
In round $m$, first select the base walk's clause using its past
history; then read and archive $\omega_m$ and perform the indexed flip.  Then choose the group $G_b$ with $b=v_2(m)$, the largest integer with $2^b$ dividing $m$, and select
its next index in cyclic order, beginning at its smallest index.
If the selected index is at most $J$, advance that auxiliary simulation
once, including its prescribed reset if needed.  Otherwise leave the
auxiliary slot unused.  On its $q$th activation, an auxiliary simulation
uses the archived symbol $\omega_q$.  Halt as soon as any simulation
satisfies the formula.

Every auxiliary activation has $q\le m$, so the required symbol is
already available.  Before its $q$th flip, its clause selector may use
only that simulation's own states and $\omega_1,\ldots,\omega_{q-1}$;
it must not inspect $\omega_q$ or any later archived symbol.  For
example, every simulation can always choose the first violated clause
in a fixed ordering.  Full assignment arrays require
$O(n\log(n+2))$ walk-state storage in addition to the archive and
any selector workspace.  Simulations
may instead be initialized lazily relative to $x_0$.

\paragraph{Why a single walk is harder.}
Suppose a single walk runs the Hamming-ball rule for its first $T_0$
symbols with radii too small for the ball to contain a satisfying
assignment, so that it cannot terminate during them (as on the
complete formulas of Theorem~\ref{thm:promise-sharp} when $2H+1<h-r$
and $n$ is so large that no assignment within distance $2H+1$ of $x_0$
satisfies $\Phi_n$), then resets to $x_0$ and continues with a radius
$H\ge h$. Let $E'$ be the survival event of the radius-$H$ phase from
$x_0$ over the remaining tape. For every word $w$ of length $T_0$, the
normalized cost $V(w)$ of \eqref{eq:reader-cover-recursion} equals
$\mathcal C_s(E')$, and since each level of the recursion multiplies
by $k\cdot k^{-s}$, the root value is
$\min\{1,k^{(1-s)T_0}\mathcal C_s(E')\}$. By
Theorem~\ref{thm:prefix-cover-robustness}, the worst-case survival
probability of the whole walk is
$\min\{1,Dk^{(1-s)T_0}\mathcal C_s(E')\}$: the decisive phase faces the
class with $Dk^{(1-s)T_0}$ in place of $D$, and a source that is
uniform during the first $T_0$ symbols and then worst-case for the
phase attains this value. The portfolio above avoids the inflation
because each copy reads the tape from its beginning, so
Theorem~\ref{thm:promise-hamming-restarts} applies to it with the
original $D$; whether a different single walk avoids it is open.

\begin{proof}[Proof of Theorem~\ref{thm:promise-radius-free-replay}]
Define counterfactual complete simulations on the full tape, even if
the combined algorithm halts earlier.  Let $\tau_{\rm base}$ be the base walk's
local stopping time and let $\tau_j$ be that of $A_j$, counting its
source-driven flips.  Each simulation reads the original tape prefix
in its original order.  Its tape law is therefore exactly $\nu$,
so Theorems~\ref{thm:promise-robust} and
\ref{thm:promise-hamming-restarts} apply to the respective simulations.
Independence between them is not needed.

Let $N_j(M)$ be the number of activations assigned to $A_j$ in the
first $M$ rounds.  If $j\in G_b$, put $z=j-(2^b-1)$ and
$q_j=2^{2b+1}$.  The activation rounds of this index are exactly
\[
 m\equiv2^b(1+2z)\pmod{q_j},
 \qquad q_j\le2(j+1)^2.
\]
Since the residue lies strictly between zero and $q_j$,
\begin{equation}\label{eq:replay-schedule-count}
 N_j(M)\ge\lfloor M/q_j\rfloor\ge M/q_j-1.
\end{equation}

By round $M$ the algorithm has made at most $M$ base flips and $M$
auxiliary source-driven flips.  In each auxiliary simulation, the
number of restored coordinates over all resets is at most the number
of its source-driven flips: each restored coordinate was flipped
during that attempt, and attempts are disjoint.  The total number of
variable updates by round $M$, or by earlier termination, is thus at
most $3M$.  In particular, $\sigma\le3\tau_{\rm base}$ pathwise.
Theorem~\ref{thm:promise-robust} therefore gives
$\E_\nu\sigma\le C\Psi_s(h+\log D+1)$.

Choose $j_* =\lceil\log_2h\rceil$.  Then $j_*\le J$ and
$h\le H_{j_*}<2h$.  If $\sigma>T$, the combined algorithm has not
terminated within the first $M=\lfloor T/3\rfloor$ rounds, since
termination in those rounds would use at most $3M\le T$ updates.
Consequently
\[
 \{\sigma>T\}\subseteq
 \{\tau_{j_*}>N_{j_*}(\lfloor T/3\rfloor)\}.
\]
Write $c_0,C_0$ for the constants in
Theorem~\ref{thm:promise-hamming-restarts}.  From
\eqref{eq:replay-schedule-count},
$N_{j_*}(\lfloor T/3\rfloor)\ge T/(3q_{j_*})-2$.
Applying the known-radius bound to this one simulation gives
\begin{align*}
 \Prb_\nu\{\sigma>T\}
 &\le C_0e^{2c_0}D\lambda^h
       \exp\!\left(-\frac{c_0T}{3q_{j_*}H_{j_*}^2}\right)\\
 &\le C_0e^{2c_0}D\lambda^h
       \exp\!\left(-\frac{c_0T}{24h^2(j_*+1)^2}\right).
\end{align*}
This proves the claimed tail after capping by one.  Summing that
capped exponential gives the second expectation bound.  Taking the
smaller of the two expectation estimates proves the result.
Finally, each round reads one new symbol and performs at least one
base update before any auxiliary work.  Hence at termination the
archive contains at most $\sigma$ symbols.
\end{proof}

\section{Proofs for restart policies}
\label{app:prefix-restarts}

\subsection{Necessity in Theorem~\ref{thm:prefix-restart-criterion}}

The proof in the main text establishes (i)$\Rightarrow$(ii), and
(ii)$\Rightarrow$(iii) is immediate. Suppose now that
$\alpha_s(x)=1$. Exact prefix capacity gives a probability law
$\mu\in\mathcal Q_{s,1}$ supported on $\mathcal N_x$.
Fix an arbitrary deterministic restart policy. At each restart, begin an
independent stream with law $\mu$. This defines an adversarial source for
the entire execution, not an assumption about the algorithm's access to
independent randomness.

For any observed word $w$, the policy determines a unique segmentation into
completed attempt words $b_1,\ldots,b_j$ and an unfinished attempt word $v$.
The resulting source therefore satisfies
\[
 \nu[w]=\mu[b_1]\cdots\mu[b_j]\mu[v]
 \le P[b_1]^s\cdots P[b_j]^sP[v]^s=P[w]^s.
\]
The formula is consistent on the prefix tree, including at restart
boundaries. Every attempt follows an original nonterminating execution
until the policy abandons it. Every prefix on which an attempt terminates
has probability zero. The countable union of these prefixes still has
probability zero, proving the converse and the stronger final assertion.

For completeness, the time bound in the sufficiency proof uses the horizon
$T$, rather than the longest failed-attempt word: a successful path can avoid
$\mathcal B$ and run longer than every word in that cover. If the longest
cover word has length $L$, the slightly sharper bound
$\Prb_\nu\{\tau_{\rm restart}>T+kL\}\le Dc^k$ also holds.

\subsection{Proof of Theorem~\ref{thm:finite-state-restart-lp}}

\paragraph{The greatest feasible vector.}
The finite-horizon controlled recursion of
Proposition~\ref{prop:controlled-prefix-value} is
\[
 V_0=\mathbf1,\qquad
 V_{t+1}(y)=\min\left\{1,\min_{f\in F_y}
                  \sum_{z\in X}R_f(y,z)V_t(z)\right\}.
\]
The sequence decreases, and its limit $V$ is a fixed point because
the state and action sets are finite. Every feasible vector $v$ in
\eqref{eq:restart-capacity-lp} satisfies $v\le V_t$ by induction,
hence $v\le V$. Conversely, the fixed point $V$ is feasible.
Thus $V$ is the componentwise greatest feasible vector and the unique
maximizer of the sum objective.

For every fixed history-dependent selector $\Lambda$, finite-horizon
optimality gives $\mathcal C_s(\{\tau_\Lambda>t\})\ge V_t(x)$.
Continuity in Theorem~\ref{thm:prefix-cover-robustness} therefore gives
\begin{equation}\label{eq:controlled-nontermination-lower}
 \mathcal C_s(\{\tau_\Lambda=\infty\})\ge V(x).
\end{equation}
This argument fixes $\Lambda$ before taking the decreasing-event limit;
it does not interchange an infimum over selectors with that limit.

\paragraph{Selecting a stationary rule.}
Set $U=\{y:V(y)<1\}$ and $S=X\setminus U$. For $y\in U$, let
\[
 I_y=\left\{f\in F_y:
              \sum_{z\in X}R_f(y,z)V(z)=V(y)\right\}.
\]
Each $I_y$ is nonempty. On $U$, solve the second linear program
\begin{equation}\label{eq:controlled-restart-potential-lp}
 \max\left\{\sum_{y\in U}\psi(y):\psi\ge0,\quad
 \psi(y)\le1+\sum_{z\in U}R_f(y,z)\psi(z)
 \ \text{for every }y\in U,\ f\in I_y\right\}.
\end{equation}
If $U$ is empty, this construction is vacuous and the selector may be
chosen arbitrarily.

The second program is feasible and bounded. Indeed, a nonzero recession
direction would be a vector $w\ge0$ satisfying
\[
 w(y)\le\sum_{z\in U}R_f(y,z)w(z)
 \qquad(y\in U,\ f\in I_y).
\]
Extend $w$ by zero on $S$. For sufficiently small $\kappa>0$,
$V+\kappa w$ remains below $\mathbf1$ on $U$ and preserves every
constraint of \eqref{eq:restart-capacity-lp}: tight constraints on $U$
use the displayed inequalities, nontight constraints have positive
slack, and constraints on $S$ can only be relaxed. This contradicts
the componentwise maximality of $V$. Thus there is no nonzero recession
direction, and the feasible polyhedron is bounded.

At an optimizer $\psi$, every state $y\in U$ has a binding action in
\eqref{eq:controlled-restart-potential-lp}. Otherwise increasing $\psi(y)$
slightly preserves its finitely many strict inequalities and relaxes
all inequalities belonging to other states, contradicting optimality.
Choose one binding action $\mu(y)\in I_y$, and choose $\mu$ arbitrarily
on $S$. Define
\[
 A(y,z)=R_{\mu(y)}(y,z)\quad(y,z\in U),\qquad
 b(y)=\sum_{z\in S}R_{\mu(y)}(y,z).
\]
The binding constraints and membership in $I_y$ give
\begin{equation}\label{eq:controlled-restart-tight-policy}
 \psi=\mathbf1+A\psi,\qquad V_U=AV_U+b.
\end{equation}
In particular $\psi\ge\mathbf1$ and $A\psi<\psi$, so $\rho(A)<1$.
Consequently,
\[
 \psi=(I-A)^{-1}\mathbf1,\qquad V_U=(I-A)^{-1}b.
\]
This second program is needed only to resolve ties properly. For
example, at a state of value zero, both an update that always terminates
and an update that always returns to the same state are tight at $V$;
the latter is not a valid choice for synthesis.

\paragraph{Exact optimality without restarts.}
Let $W$ be the nontermination capacity vector of the fixed stationary
selector $\mu$. Equivalently, $W$ is the greatest feasible vector for
the fixed-matrix constraints $0\le w\le\mathbf1$, $w\le R_\mu w$,
by the same decreasing-recursion argument. Since $V$ is feasible for
these constraints, $W\ge V$ and hence $W_S=\mathbf1$.
On $U$,
\[
 W_U\le AW_U+b,
 \qquad
 W_U\le(I-A)^{-1}b=V_U.
\]
Thus $W=V$. Combining this with
\eqref{eq:controlled-nontermination-lower} and the exact capacity
identity proves that $\mu$ attains the stated minimum, simultaneously
for every initial state and every $D\ge1$.

\paragraph{Restart tail bound.}
Suppose $x_0\in U$. Under $\mu$ and immediate restart on entry to $S$,
the powered live matrix on $U$ is
\[
 B=A+b e_{x_0}^{\mathsf T}.
\]
Set
\[
 g=\psi+\frac{\psi(x_0)}{1-V(x_0)}V_U,
 \qquad H=\max_{y\in U}g(y).
\]
Equation~\eqref{eq:controlled-restart-tight-policy} gives
$Bg=g-\mathbf1$. Therefore $g\ge\mathbf1$, $\rho(B)<1$, and
$g=(I-B)^{-1}\mathbf1$. If $H>1$, then
$Bg\le(1-H^{-1})g$; if $H=1$, then $B\mathbf1=0$ and termination
occurs after one source-consuming step.

Summing powered probabilities over complete surviving words gives
\begin{align}
 \Prb_\nu\{\tau_{\rm restart}>t\}
 &\le\min\{1,D e_{x_0}^{\mathsf T}B^t\mathbf1\},
 \label{eq:restart-tail}\\
 \mathbb E_\nu\tau_{\rm restart}
 &\le Dg(x_0)
 =D\,\frac{\psi(x_0)}{1-V(x_0)}
 =D\,\frac{e_{x_0}^{\mathsf T}(I-A)^{-1}\mathbf1}{1-V(x_0)}.
 \label{eq:restart-mean}
\end{align}
For $H>1$, the first bound is at most
$Dg(x_0)(1-H^{-1})^t$. The mean bound follows by summing the
matrix bound over $t\ge0$. These estimates invoke the source
constraint only on prefixes of the whole execution tape; restarting
does not restore a conditional source guarantee.

\paragraph{Impossibility for arbitrary controlled restart policies.}
Suppose $V(x_0)=1$, and fix any deterministic nonanticipating policy
that chooses updates and may restart to $x_0$ between successive steps.
We construct one source in $\mathcal Q_{s,1}$ for this entire policy.
Write $V=0$ at terminated states. At a positive-probability history
whose current live state is $y$ and chosen update is $f$, set
\[
 L(y,f)=\sum_a\pi(a)^s V(\mathcal T_f(y,a))\ge V(y)>0,
 \qquad
 q(a\mid\text{history})=
       \frac{\pi(a)^s V(\mathcal T_f(y,a))}{L(y,f)}.
\]
These probabilities sum to one. They assign zero probability to
termination and to states of value zero. Within any attempt, their
product along a word $b$ ending at $z$ is at most
$P[b]^sV(z)/V(x_0)=P[b]^sV(z)$, by telescoping the successive
values and using $L(y,f)\ge V(y)$.

For any prefix $w$ of the execution tape, the policy determines its
unique segmentation into completed attempts and a current partial
attempt.
Multiplying the preceding bounds over these segments gives
$\nu[w]\le P[w]^s$, since each endpoint value is at most one.
Thus the constructed law lies in $\mathcal Q_{s,1}$ and has
nontermination probability one. Its definition may depend on the
fixed policy, as permitted by the minimax quantifiers. This proves
necessity even for policies that retain their entire histories across
restarts.

\paragraph{Exact computation.}
If the powered coefficients $R_f(y,z)$ are rational, the first linear
program, its unique optimizer, the tight action sets, and the second
linear program all have polynomial-size rational descriptions.
Exact linear programming and rational linear algebra therefore compute
$\mu$, $S$, and $g$ in polynomial bit complexity in the explicit
state--action input. The resulting execution-time constants need not
be polynomially bounded. Approximate representations of arbitrary real
values $\pi(a)^s$ do not by themselves provide the exact rational input
required by this assertion.

\paragraph{Rounded weights.}
Replacing each $\pi(a)^s$ by a rational upper bound enlarges the
feasible set of \eqref{eq:restart-capacity-lp} and every cover sum in
the proof, so the selector and restart rule computed from the rounded
weights are valid for $\mathcal Q_{s,D}$: the selector's worst-case
nontermination probability is at most $\min\{1,D\bar V(x)\}$ for the
computed value $\bar V\ge V$, and, whenever $\bar V(x_0)<1$, the tail
bound of (ii) holds with the computed constants. Rounding does not
preserve exactness: the equalities in (i), the optimality of the
selector, and the impossibility statements refer to the true weights.
Each finite-horizon value is continuous in the weights and $V$ is their
decreasing limit, so $\bar V$ converges to $V$ from above as the
rounding error tends to zero: the loss in (i) vanishes in the limit,
and $V(x_0)<1$ is certified at finite precision, as needed at an
irrational $s$ such as $s_*$.

\subsection{Adaptive versus predetermined restarts}

\begin{proposition}[Adaptive versus predetermined restarts]
\label{prop:restart-oblivious-separation}
A three-live-state procedure has reference mean running time $8/3$ under a
uniform four-symbol source and the following properties. For every
predetermined sequence of positive integer cutoffs, allowing a final infinite
cutoff, some source in $\mathcal Q_{1/2,1}$ gives nontermination probability
at least $1/2$. In contrast, one stationary restart rule satisfies, for every
$D\ge1$ and integer $T\ge2$,
\[
 \max_{\nu\in\mathcal Q_{1/2,D}}
 \Prb_\nu\{\tau_{\rm adaptive}>T\}=\min\{1,D\,T2^{1-T}\}.
\]
\end{proposition}

\begin{proof}

Use source alphabet $\{0,1,2,3\}$, initial state $x$, and the transition
table below, where $\dagger$ denotes termination.
\[
\begin{array}{c|cccc}
 &0&1&2&3\\ \hline
 x&\dagger&y&z&z\\
 y&\dagger&y&y&y\\
 z&z&\dagger&\dagger&\dagger
\end{array}
\]
Under the uniform product law,
\[
 \mathbb E_P\tau
 =1+\tfrac14\cdot4+\tfrac12\cdot\tfrac43=\tfrac83.
\]
At exponent $1/2$, each symbol has weight $1/2$. The original powered live
matrix, in state order $(x,y,z)$, is
\[
 R=\begin{pmatrix}0&1/2&1\\0&3/2&0\\0&0&1/2\end{pmatrix}.
\]
For every $T\ge1$, the failed length-$T$ words $F_T$ have total weight
\[
 Z_T^{\rm live}:=\sum_{w\in F_T}P[w]^{1/2}
 =\tfrac12(3/2)^{T-1}+(1/2)^{T-1}\ge1.
\]
For $T=1,2$, the values are $3/2,5/4$; for $T\ge3$, the first term alone
exceeds one. The original nontermination capacities are $(1/2,1,0)$.

For a finite cutoff $T$, put the probability law
$q_T(w)=P[w]^{1/2}/Z_T^{\rm live}$ on $F_T$. At a nonempty partial word
$u$ of length $j\le T$, the ratio $q_T[u]/P[u]^{1/2}$ is zero or one of
\[
 \frac{(3/2)^{T-j}}{Z_T^{\rm live}}\le2
 \quad\text{on the $y$ branch},\qquad
 \frac{(1/2)^{T-j}}{Z_T^{\rm live}}\le1
 \quad\text{on the $z$ branch}.
\]
At a complete block, this ratio is $1/Z_T^{\rm live}\le1$.
For any predetermined finite-cutoff sequence, concatenate independent blocks
with these laws. Each complete block contributes a ratio at most one, and
the final partial block contributes at most two. The resulting law
$\mu\in\mathcal Q_{1/2,2}$ is supported on nontermination of the scheduled
procedure. If the schedule has a final infinite cutoff, choose its first
symbol to be $1$, then choose independent uniform symbols from
$\{1,2,3\}$ forever. The partial-prefix ratios in this final attempt are
also at most two. Hence $\mu/2$ is a subprobability measure satisfying the
unit prefix caps. Prefix-flow completion extends it to a probability law
in $\mathcal Q_{1/2,1}$ while retaining nontermination mass at least $1/2$.
The adversarial source may depend on the predetermined schedule.

The adaptive rule restarts immediately on entry to $y$. Its powered live
matrix on $(x,z)$ is
$B=\left(\begin{smallmatrix}1/2&1\\0&1/2\end{smallmatrix}\right)$.
The exact capped recursion gives $v_T(z)=2^{-T}$ and
$v_0(x)=v_1(x)=v_2(x)=1$. For $T\ge3$,
\[
 v_T(x)=\tfrac12v_{T-1}(x)+2^{-(T-1)}=T2^{1-T}.
\]
The formula also holds at $T=2$. Exact prefix capacity now gives
$\max_{\nu\in\mathcal Q_{1/2,D}}
\Prb_\nu\{\tau_{\rm adaptive}>T\}=\min\{1,D\,T2^{1-T}\}$.
\end{proof}

\section{Proofs for scheduling}\label{app:scheduling-proofs}
\subsection{Optimal priority for heterogeneous retry tasks}
\label{app:retry-priority-proof}
\begin{proof}[Proof of Theorem~\ref{thm:retry-priority}]
For a set $S$ of unfinished tasks, the controlled recursion gives
\begin{equation}\label{eq:retry-priority-bellman}
 V_{t+1}(S)=\min\left\{1,\min_{i\in S}
 \bigl((W-w_i)V_t(S)+w_iV_t(S\setminus\{i\})\bigr)\right\},
\end{equation}
with $V_t(\varnothing)=0$ and $V_0(S)=1$ for nonempty $S$.
Symbols with the same success status lead to the same state, so their
weights add even when a selector remembers the individual labels.
Proposition~\ref{prop:controlled-prefix-value} gives the exact optimum
$\min\{1,DV_T(S)\}$.

We first record three monotonicities. Values decrease with $t$, because
the recursion operator is monotone and $V_1\le V_0$.
They increase when tasks are added. Indeed, suppose $R\subseteq S$ and
argue by induction on $t$. An action $i\in R$ has at least its
corresponding cost on $R$. An action $i\in S\setminus R$ has cost at
least $W V_t(R)\ge V_t(R)\ge V_{t+1}(R)$, since $W\ge1$.
Taking the minimum and truncating at one preserve the comparison.

For fixed $W$, increasing a task's weight cannot increase the value.
To see this, compare recursions in which its weight is $a$ and
$b\ge a$. By induction, the present-state values satisfy $X_a\ge X_b$,
and the common value after completing that task is $Y\le X_b$.
The corresponding action costs differ by
\[
 (W-a)X_a+aY-\bigl((W-b)X_b+bY\bigr)
 =(W-a)(X_a-X_b)+(b-a)(X_b-Y)\ge0.
\]
Other action costs are ordered by induction. These arguments apply to
the abstract recursion for all $0<w_i\le W$ and $W\ge1$, without
requiring a modification of the reference source.

Subtracting the costs of two possible first tasks cancels the common
$W$ term. What remains is a comparison of the values obtained by
removing either task. This motivates the following stronger invariant,
which we prove by induction on $t$: for distinct
$i,j\in S$ with $a=w_i<b=w_j$,
\begin{equation}\label{eq:retry-priority-invariant}
 (b-a)V_t(S)
 \le bV_t(S\setminus\{j\})-aV_t(S\setminus\{i\}).
\end{equation}
It holds at $t=0$. The inequality says exactly that the action for $i$
has no larger cost than the action for $j$ in
\eqref{eq:retry-priority-bellman}. Equal-weight tasks are interchangeable.
Thus, if the invariant holds at time $t$, a minimum-weight task is an
optimal first action for horizon $t+1$.

We shall use the following elementary truncation fact. If
$0\le z\le y\le x$ and $(b-a)x\le by-az$, then
\begin{equation}\label{eq:retry-priority-clipping}
 (b-a)\min\{1,x\}
 \le b\min\{1,y\}-a\min\{1,z\}.
\end{equation}
If $x\le1$, nothing changes. If $y\ge1$, the right side is at least
$b-a$. Otherwise $x>1>y$, and the unchanged right side is at least
$(b-a)x\ge b-a$.

Fix $i,j$ as above, and choose a minimum-weight task $k\in S$.
First suppose $k\notin\{i,j\}$ and put $c=w_k$.
The optimal first action in $S$, $S\setminus\{i\}$, and
$S\setminus\{j\}$ can be $k$. Before truncation, each value is the
same nonnegative linear combination, with coefficients $W-c,c$, of
the corresponding values with and without $k$ at time $t$.
Combine \eqref{eq:retry-priority-invariant} for $S$ and
$S\setminus\{k\}$ to obtain the desired inequality before truncation.
The three quantities have order $x\ge y\ge z$: adding a task increases
the value, and replacing a weight $a$ by $b$ decreases it.
Equation~\eqref{eq:retry-priority-clipping} proves the invariant after
truncation.

The remaining case is $k=i$, since $a<b$ excludes $k=j$. Write
\[
 X=V_t(S),\quad Y=V_t(S\setminus\{j\}),\quad
 Z=V_t(S\setminus\{i\}),\quad H=V_t(S\setminus\{i,j\}).
\]
The monotonicities give $X\ge Y\ge Z\ge H$. Set
\[
 x=(W-a)X+aZ,\qquad y=(W-a)Y+aH,\qquad z=(W-b)Z+bH.
\]
The first two quantities are the optimal costs before truncation for
$S$ and $S\setminus\{j\}$. The third is a permissible cost for
$S\setminus\{i\}$, obtained by choosing $j$. Hence its truncated
value is an upper bound on $V_{t+1}(S\setminus\{i\})$.
We have $x\ge y$ and
\[
 y-z=(W-a)(Y-Z)+(b-a)(Z-H)\ge0.
\]
Furthermore,
\[
 by-az-(b-a)x=(W-a)\bigl(bY-aZ-(b-a)X\bigr)\ge0.
\]
Apply \eqref{eq:retry-priority-clipping}, then replace its last truncated
quantity by the no-larger actual value. This proves
\eqref{eq:retry-priority-invariant} at $t+1$ and completes the induction.
Thus increasing-weight priority is optimal at every horizon. Following
that order leaves only suffixes of the sorted list and gives the stated
$O(nT)$ recursion. Monotonicity of $x\mapsto\min\{1,Dx\}$ proves
simultaneous optimality for every $D\ge1$.

Finally, fix $\theta>0$, let $W_\theta=\sum_a\pi(a)^\theta$, and
consider the tilted law $\pi_\theta(a)=\pi(a)^\theta/W_\theta$. Under
its independent stream, preassign
independent streams to the tasks. The number of attempts for task $i$
is geometric with success probability $\sum_{a\in B_i}\pi_\theta(a)$.
The workloads are independent and finite almost surely. Every permitted
selector uses exactly their sum: before all tasks complete it must
attempt some unfinished task, and none can receive more than its
preassigned workload. The next unrevealed symbol of any selected stream
has law $\pi_\theta$, even when the selection uses all previous observations.
The stopping-time distribution is therefore selector independent.
Multiplication by $W_\theta^T$ proves equality of the prefix power sums
$\widetilde Z_T(\theta)$, and the identity
$Z_T(\theta)=W_\theta\widetilde Z_{T-1}(\theta)$ proves it for
$Z_T(\theta)$.
\end{proof}

\subsection{Equal stopping laws and different robustness}
\label{app:repair-order-proof}
\begin{proof}[Proof of Theorem~\ref{thm:repair-order-multiscale}]
Put
\[
 u=\sqrt{\pi(0)}=\frac{\sqrt{15}}4,
 \qquad v=\sqrt{\pi(1)}=\frac14.
\]
We first construct $\nu_n$. Through time $8n$, generate actions with
probabilities $3/4$ for $0$ and $1/4$ for $1$ until the $n$th zero
appears; after that zero, output only zeros through time $8n$. After
time $8n$, generate independent actions according to $P$. This specifies
one law on tapes, without reference to the selector.

A prefix of length at most $8n$ preceding the $n$th zero, with $k$ zeros, satisfies
\[
 \frac{\nu_n[w]}{P[w]^{1/2}}
   =\left(\frac{3}{\sqrt{15}}\right)^k\le1.
\]
Suppose the $n$th zero appears at time $L\le8n$. For a positive-mass
prefix of length $t\in[L,8n]$, the corresponding ratio is
\[
 \left(\frac{3}{\sqrt{15}}\right)^n u^{-(t-L)}
 \le\left(\frac{3/4}{u^8}\right)^n<1,
\]
since $L\ge n$ and $u^8=(15/16)^4>3/4$. Zero-mass prefixes satisfy
the required bound automatically. After time $8n$, each new reference
action multiplies the ratio by $\pi(a)^{1/2}\le1$. Thus every finite
prefix satisfies the defining inequality for $\mathcal Q$. Under
$\Lambda_A$, either fewer than $n$ zeros have appeared by time $8n$, so
an $A$ task remains unfinished, or the $B$ task is attempted only after
the $n$th zero and subsequently sees zeros alone. In both cases
$\tau>8n$, proving \eqref{eq:repair-order-bad}.

For \eqref{eq:repair-order-good}, consider an arbitrary
$\nu\in\mathcal Q$ and use $\Lambda_B$. If no $1$ appears during the
first $4n$ actions, the tape begins with $0^{4n}$, an event of
$\nu$-probability at most
\[
 u^{4n}=\left(\frac{225}{256}\right)^n.
\]
Otherwise let $j\le4n$ be the first position containing $1$. This
completes $B$, and survival after $8n$ repairs requires fewer than $n$
zeros among the remaining $\ell=8n-j\ge4n$ actions. Set
$r=v/(3u)=1/(3\sqrt{15})$. Their total power weight is at most
\begin{align*}
 \sum_{k=0}^{n-1}\binom{\ell}{k}u^kv^{\ell-k}
 &\le r^{-(n-1)}(v+ur)^\ell\\
 &\le r^{-(n-1)}(1/3)^{4n}
   =r\left(\frac{\sqrt{15}}{27}\right)^n.
\end{align*}
The first-$1$ prefix $0^{j-1}1$ has power weight $u^{j-1}v$.
Applying the cylinder bound and summing over $j\le4n$ gives the claimed
second term because
\[
 \sum_{j=1}^{4n}u^{j-1}v\le\frac{v}{1-u}=4+\sqrt{15}.
\]
This argument holds for every source in $\mathcal Q$, including the
source $\nu_n$ used above.
To obtain the simpler constant in \eqref{eq:repair-order-good}, put
$\alpha=\sqrt{15}/27$, $\beta=225/256$, and
$c=(4+\sqrt{15})/(3\sqrt{15})$.  The bound is
$\beta^n+c\alpha^n$.  Since $\alpha<\beta$ and
$c\alpha/\beta=256(4+\sqrt{15})/18225<1/8$, it is at most
$(9/8)\beta^n<1$ for every $n\ge1$.

For the lower bound in \eqref{eq:repair-order-good}, take an i.i.d.
source with probabilities $u$ and $1-u$.  Both are at most the
corresponding square-root reference probabilities, since $1-u\le v$.
It belongs to $\mathcal Q$, and its prefix $0^{8n}$ leaves $B$
unfinished under every selector.  Its probability
$u^{8n}=(15/16)^{4n}$ is therefore a lower bound for every selector,
including the minimax optimum.

The selector invariance of the stopping law and all prefix power sums
follows from Theorem~\ref{thm:retry-priority}, applied to these success sets.
\end{proof}

\section{Proof for noisy predecessor search}
\label{app:noisy-predecessor}

\begin{proof}[Proof of Theorem~\ref{thm:noisy-predecessor}]
Set $x_0=-\infty$ and $x_{n+1}=+\infty$.
Leaf $i$ represents $[x_i,x_{i+1})$ and returns index $i$, with zero
meaning no predecessor. Internal nodes represent intervals $[x_i,x_j)$
and split at $x_k$, $i<k<j$. Chain nodes inherit their leaf's interval
and label. These half-open intervals handle equality using strict queries.

At a nonterminal state with interval $[x_i,x_j)$, use the lower
and upper comparisons $x<x_i$ and $x<x_j$ when their boundaries are
finite. Verification requires the lower answer to be false and the
upper answer to be true. An infinite boundary is known exactly;
replace its comparison by a dummy comparison $x<x_1$ and ignore the answer.
At an original internal node the third comparison is its split comparison
$x<x_k$; at an original leaf or chain node it is another dummy
$x<x_1$. Each of these three logical comparisons is answered by the
majority of exactly $\ell$ repetitions of the same primitive query.
All $3\ell$ queries are made before the move, including dummy queries
and queries after an earlier majority already makes verification fail.
On failure move to the parent, with a requested upward move at the
root replaced by staying there. On success take the split-selected
child at an internal node and the unique downward chain edge otherwise.
At a chain endpoint return immediately, without further queries.

Let $z$ be the correct original leaf and $z^*$ its chain endpoint.
If a node's interval excludes $x$, three correct majorities make
verification fail, and its parent is the neighbor toward $z^*$.
If its interval contains $x$, correct answers select the child toward
$z^*$, or advance on $z$'s chain. Thus every raw block in $G$ moves
one edge toward $z^*$. Every other block moves at most one edge.
For live states use the positive integer potential
$d(v)=\operatorname{dist}(v,z^*)$, initially
$d_0=2h+L_\eta$; assign potential zero on termination.

For odd $\ell$, define the weight of a correct majority by
\[
 M_s=\sum_{j>\ell/2}\binom\ell j
             p^{sj}(1-p)^{s(\ell-j)}.
\]
Since $p>1/2$, the law of large numbers ensures an odd $\ell$ with
$M_1>2^{-1/3}$. Choose the least such $\ell$, so it depends only on $p$.
Group raw errors into consecutive blocks of length $3\ell$, with
reference mass $\pi(e)=p^{3\ell-|e|_1}(1-p)^{|e|_1}$.
Equation~\eqref{eq:noisy-search-prefix} directly bounds every
length-$T$ block prefix by $D$ times its reference mass to power $s$;
no change in the primitive source assumption is made.
Summing over raw patterns gives
\[
 a_s=M_s^3,\qquad
 b_s=(p^s+(1-p)^s)^{3\ell}-M_s^3.
\]
Both are sums of $\pi(e)^s$ over fixed nonempty sets, so both strictly
decrease in $s$. At $s=0$, their product is $7\cdot2^{6\ell-6}>1/4$;
at $s=1$, it is $M_1^3(1-M_1^3)<1/4$.
This proves existence and uniqueness of $s_c$.
Also $a_{s_c}\ge M_1^3>1/2$, so
$a_{s_c}=\lambda/2$, $b_{s_c}=1/(2\lambda)$, and $\lambda>1$.
This standard amplification is analyzed on raw patterns: the correct
weight is $M_s^3$, not $M_1^{3s}$. No source condition on the
majority answers is assumed.

For every nonnegative increasing $f$ with $f(0)=0$, the live children
of a history of potential $d$ satisfy
\begin{equation}\label{eq:noisy-search-weighted-comparison}
 \sum_{e:\,\mathrm{live}}\pi(e)^{s_c} f(d_e)
 \le a_{s_c}f(d-1)+b_{s_c}f(d+1).
\end{equation}
This is the complete one-step input to the proof of
Lemma~\ref{lem:promise-potential-barrier}: its tree induction uses
powered edge masses, so the same proof applies to these nonuniform
edges. In particular, for any integer $m>d_0$, the weight of
prefixes first reaching $m$ before termination is at most
\[
 H_m(d_0)=\frac{d_0}{m}\lambda^{d_0-m},
 \qquad H_m(d)=\frac dm\lambda^{d-m}.
\]
Indeed, $H_m(0)=0$, $H_m(m)=1$, and
$a_{s_c}H_m(d-1)+b_{s_c}H_m(d+1)=H_m(d)$;
expanding the finite tree to any depth bounds its first-hit prefixes
by $H_m(d_0)$, and increasing the depth gives the infinite-horizon
bound. The remaining live prefixes confined below $m$ have powered
mass at most
\[
 \frac{\lambda^{d_0}}{\sqrt{\lambda^2-1}}
 \exp\left(-\frac{\pi^2T}{2(2m-1)^2}\right),
\]
by the same reflected-strip comparison and diagonal conjugation in
that lemma. Multiplying the two cover sums by $D$ and optimizing $m$
therefore gives
\[
 \Prb_\nu\{Q>3\ell T\}
 \le\min\{1,C_pD\lambda^{d_0}e^{-c_pT^{1/3}}\},
 \qquad \E_\nu Q\le C_p(d_0+\log D+1)^3.
\]
The inclusion of source classes gives this bound for all $s\ge s_c$.
For $s>s_c$, since $\max_e\pi(e)=p^{3\ell}$,
\[
 \frac{a_s}{\lambda}+b_s\lambda
 \le p^{3\ell(s-s_c)}
 \left(\frac{a_{s_c}}{\lambda}+b_{s_c}\lambda\right)
 =p^{3\ell(s-s_c)}.
\]
Finite-tree induction on the powered sum of $\lambda^d$ over live
depth-$T$ prefixes yields
$\Prb_\nu\{Q>3\ell T\}\le\min\{1,D\lambda^{d_0}p^{3\ell(s-s_c)T}\}$.
Summing and counting $3\ell$ queries per move gives
$C_p(d_0+\log D+1)/(s-s_c)$.
Since $d_0=2h+L_\eta=O_p(K)$, these bounds prove the query claim and
almost-sure termination.

Finally consider an incorrect leaf $u$ of depth $d_u$, and let $c$
be the depth of its lowest common ancestor with $z$.
Its chain endpoint $u^*$ satisfies
\[
 \operatorname{dist}(u^*,z^*)
 =2h+2d_u-2c+2L_\eta\ge d_0+L_\eta+2.
\]
Set $m=d_0+L_\eta+1$. Any run terminating at $u^*$ must first reach
distance $m$ at a live state: its final endpoint has distance at
least $m+1$, and each move changes distance by at most one.
Consequently the first-hit bound, applied before resetting the
terminal potential to zero, gives
\[
 \Prb_\nu\{\text{incorrect predecessor}\}
 \le D H_m(d_0)
 \le D\lambda^{-L_\eta-1}\le\eta.
\]
Only the algorithm's visited leaf label is returned; the unknown
correct endpoint is used solely in this analysis.
\end{proof}

\section{Composition under a shared source}
\label{app:component-proof}
Sharing a random stream can alter robustness even when a selector may
interleave components arbitrarily. The following process gives an exact
comparison between a prefix probability bound on the shared tape and
separate bounds on each component's subsequence.

Each component is a connected $N$-vertex graph with binary vertex
values. Select an edge whose endpoints disagree. Four source labels
replace its endpoint values by $00$, leave them unchanged, reverse
them, or replace them by $11$. The reference probabilities are
$(49/100,1/4,1/4,1/100)$, respectively. Constant components are finished
and never selected again. This extends the discordant voting model of
Cooper, Dyer, Frieze, and Rivera~\citep{cooper-et-al-2018-discordant}.

Let $\Delta_N\in(1/2,1)$ be the unique solution of
\[
 2\,4^{-\Delta_N}+2(7/100)^{\Delta_N}\cos(\pi/N)=1.
\]

\begin{corollary}[Small individual failure, certain joint failure]
\label{cor:composition-main}
For every $N\ge3$, start each of $N-1$ components with exactly one vertex
equal to one and all others zero. In isolation, each has worst-case nontermination probability
\[
 \max_{\nu\in\mathcal Q_{\Delta_N,1}}\nu\{\tau_i=\infty\}
 =7^{-\Delta_N(N-2)}=\exp[-\Theta(N)].
\]
There is one source in $\mathcal Q_{\Delta_N,1}$ under which their
combined execution never terminates, simultaneously for every
deterministic nonanticipating rule that chooses a disagreeing edge.
The rule may interleave components arbitrarily and remember all past
labels. Under the independent reference source, every such rule has
expected total running time at most $25(N-1)/12$.
\end{corollary}

The following formula proves the corollary and also allows components
of different sizes and arbitrary initial assignments.

\begin{theorem}[Exact robustness under arbitrary interleaving]
\label{thm:modular-disagreement-capacity}
Fix $N\ge3$, put $s=\Delta_N$, $\theta=\pi/N$, and $r=7^s$, and define
\[
 V_N(k)=
 \begin{cases}
  0,&k=0,\\
  r^{k-N+1}\dfrac{\sin(k\theta)}{\sin\theta},&1\le k<N-1,\\
  1,&k\ge N-1.
 \end{cases}
\]
Let the graph be a disjoint union of finitely many connected components,
each with at least $N$ vertices. Let $K_0$ count the ones in the initially
nonconstant components. For every deterministic nonanticipating edge
selector and every $D\ge1$,
\begin{equation}\label{eq:modular-disagreement-capacity}
 \max_{\nu\in\mathcal Q_{s,D}}\nu\{\tau=\infty\}
 =\min\{1,DV_N(K_0)\}.
\end{equation}
One source depending only on $N,K_0,D$ attains this value simultaneously
for every allowed graph, initial assignment with count $K_0$, and
selector. Under the independent reference source,
\begin{equation}\label{eq:modular-nominal-time}
 \mathbb E_P\tau\le\frac{25}{12}K_0.
\end{equation}
\end{theorem}

\begin{proof}
Write $V=V_N$ and let
\[
 d=(49/100)^s,
 \qquad b=2(1/4)^s,
 \qquad u=(1/100)^s,
 \qquad a=\sqrt{du}=(7/100)^s.
\]
The equation defining $s$ is $b+2a\cos\theta=1$. Its left side is
strictly decreasing in $s$, exceeds one at $s=1/2$, and is below one
at $s=1$. Thus $1/2<s<1$. In particular, $b<1$,
$d=ar$, $u=a/r$, and $r>\sqrt7>2\cos\theta$.

Set $h(k)=r^k\sin(k\theta)$ for $0\le k\le N$.
The sine addition formula gives
\[
 dh(k-1)+bh(k)+uh(k+1)=h(k),\qquad 1\le k<N.
\]
Moreover, $h$ is strictly increasing on $1,\ldots,N-1$:
for $1\le k\le N-2$,
$\sin(k\theta)\le2\cos\theta\sin((k+1)\theta)$ and
$r>2\cos\theta$. Hence $0<V(k)<1$ for $1\le k<N-1$,
and the same harmonic identity holds for $V$ there. At the upper
boundary it gives
\begin{equation}\label{eq:composition-boundary-identity}
 dV(N-2)+b=1.
\end{equation}

At any history, let $K$ count the ones in components that are still
nonconstant. A downward action sends $K$ to $K-1$; either neutral
action leaves it unchanged. An upward action sends $K$ to $K+1$,
unless it makes an $L$-vertex component all ones, in which case
$K$ becomes $K+1-L$. Such a completion requires
$K\ge L-1\ge N-1$. Thus below $N-1$ the count has the same
increments for every selector, even though this need not hold above
that level.

For the upper bound, only $1\le K_0<N-1$ requires proof.
Cover each nonterminating tape at its first visit to $K=N-1$, or
at depth $T$ if it has not yet visited that level. Before the first
visit, the powered count matrix on $1,\ldots,N-2$ is tridiagonal
with lower diagonal $d$, diagonal $b$, and upper diagonal $u$.
Call it $A$. Diagonal similarity to a symmetric tridiagonal matrix gives
\[
 \rho(A)=b+2a\cos\frac{\pi}{N-1}<1.
\]
The harmonic identity with boundary values $V(0)=0$ and $V(N-1)=1$
therefore makes the total weight of first visits to $N-1$
equal to $V(K_0)$. The remaining live prefixes at depth $T$ have
weight $e_{K_0}^{\mathsf T}A^T\mathbf1$, which tends to zero.
Applying the prefix bound to this finite cover and letting $T\to\infty$
gives $\nu\{\tau=\infty\}\le\min\{1,DV(K_0)\}$.

For the matching source, suppose $K_0>0$. An auxiliary law $q$
maintains a counter $Y_0=K_0$. Let $\delta(z)\in\{-1,0,1\}$ be
the increment assigned to action label $z$. When $1\le Y_t<N-1$, use
\[
 q(z\mid Y_t)=\pi(z)^s\frac{V(Y_t+\delta(z))}{V(Y_t)}.
\]
When $Y_t\ge N-1$, give each neutral label probability $(1/4)^s$,
the downward label probability $1-b$, and the upward label probability
zero. Update $Y_{t+1}=Y_t+\delta(z)$.
The harmonic identity normalizes these probabilities below $N-1$,
and normalization above that level is immediate. The downward
probability at one is zero, so $Y_t$ stays positive.

Every upward move starts below $N-1$, so it cannot finish a component.
A downward completion produces an all-zero component and has no
additional effect on the count. Consequently $K_t=Y_t$ for every
graph, initial assignment, and selector under consideration. Every
tape in the support of $q$ is nonterminating for all these processes.

Equation~\eqref{eq:composition-boundary-identity}, together with
$V(N-2)<1$, gives $1-b=dV(N-2)<d$. It follows in the upper range
as well as the lower range that
\[
 q(z\mid Y_t)\le
 \pi(z)^s\frac{V(Y_{t+1})}{V(Y_t)}.
\]
Multiplying along a prefix yields
\begin{equation}\label{eq:component-common-source-envelope}
 q[w]\le P[w]^s\frac{V(Y_{|w|})}{V(K_0)}
 \le\frac{P[w]^s}{V(K_0)}.
\end{equation}
Put $p=\min\{1,DV(K_0)\}$. The subprobability law $\eta=pq$
has mass $p$ and obeys $\eta[w]\le DP[w]^s$.
It extends to a probability law $\nu$ obeying the same bound:
at each node, distribute its assigned mass among the children,
preserving their $\eta$ masses and respecting their capacities
$DP[wz]^s$. This is possible because the lower bounds sum to
$\eta[w]$, whereas the capacities sum to
$DP[w]^s\sum_z\pi(z)^s\ge DP[w]^s$.
Starting with root mass one and using a fixed label order gives an
extension depending only on $N,K_0,D$. Its common nontermination
mass is at least $p$, proving equality in
\eqref{eq:modular-disagreement-capacity}. If $K_0=0$, all components
are already constant and the assertion is immediate. When
$K_0\ge N-1$ and $D=1$, $p=1$ and $q$ itself is the required source.

Finally, under $P$, a live update has downward drift at least
$49/100-1/100=12/25$ in $K$; finishing an all-one component only
increases that downward drift. Summing the conditional drift bound
up to $\tau\wedge t$ gives
$(12/25)\mathbb E_P(\tau\wedge t)\le K_0$.
Letting $t\to\infty$ proves~\eqref{eq:modular-nominal-time}.
\end{proof}

For one component starting with one vertex equal to one,
$V_N(1)=7^{-\Delta_N(N-2)}$. For $N-1$ such components,
$V_N(N-1)=1$ and the same source forces nontermination under every
interleaving. Equation~\eqref{eq:modular-nominal-time} gives the
reference mean bound. This proves Corollary~\ref{cor:composition-main}.

\bigskip
\Needspace{18\baselineskip}
\phantomsection
\addcontentsline{toc}{section}{AI methodology statement}
\section*{AI methodology statement}
The author posed the research questions, chose the
source model and the three settings, and directed the investigation
step by step, selecting at each stage which of several candidate
directions to pursue; this led successively to the scheduling
separation, the sharp threshold and cubic rate for the SAT walk, the
restart criterion, and the noisy-search guarantee. The author checked,
selected, and organized the results and supplied their interpretation
as one source model, three applications (scheduling, local SAT search,
and noisy predecessor search), and one decision available in all of
them, when to restart. Throughout this process the author used ChatGPT
and Codex (OpenAI) as research aids; their assistance included
exploring and drafting proof arguments, comparing related literature,
and revising exposition, and it materially influenced
Sections~\ref{sec:prefix-cover-robustness}--\ref{sec:prefix-restarts}
and the associated appendices. Claude (Anthropic) assisted with the
organization and exposition of  the main text and the
appendices, including the wording and structure of several theorem
statements and proofs. The author
evaluated the proposed arguments, checked the proofs and cited sources,
revised the material retained in the paper, and made the final
decisions about its claims and presentation. The author verified the
correctness and originality of all content, including references, and
takes full responsibility for the paper's accuracy, attribution, and
integrity.

\clearpage
\phantomsection
\addcontentsline{toc}{section}{References}
\newcommand{\etalchar}[1]{$^{#1}$}

\end{document}